\pdfoutput=1
\documentclass[12pt,reqno]{amsart}

\numberwithin{equation}{section}
\usepackage[tt=false]{libertine}
\usepackage{mathtools}

\usepackage{amssymb,mathrsfs}
\usepackage[varbb]{newpxmath}

\let\savedbigtimes\bigtimes
\let\bigtimes\relax
\usepackage{mathabx} 
\let\bigtimes\savedbigtimes

\usepackage[margin=1in]{geometry}
\usepackage{graphicx}
\usepackage{enumerate}
\usepackage{bbm}
\usepackage{hyperref,color}
\usepackage[capitalize,nameinlink]{cleveref}
\usepackage[dvipsnames]{xcolor}
\hypersetup{
	colorlinks=true,
	pdfpagemode=UseNone,
    citecolor=OliveGreen,
    linkcolor=NavyBlue,
    urlcolor=black,
	pdfstartview=FitW
}
\usepackage{appendix}
\crefname{appsec}{Appendix}{Appendices}
\usepackage{tikz}
\usepackage{bm}
\usepackage{mathtools}

\newtheorem{theorem}{Theorem}[section]

\newtheorem{lemma}[theorem]{Lemma}

\theoremstyle{definition}
\newtheorem{definition}[theorem]{Definition}

\newtheorem*{assumption*}{Assumption}
\newtheorem{remark}[theorem]{Remark}

\crefname{lemma}{Lemma}{Lemmas}
\crefname{theorem}{Theorem}{Theorems}
\crefname{definition}{Definition}{Definitions}
\crefname{fact}{Fact}{Facts}
\crefname{claim}{Claim}{Claims}
\crefname{proposition}{Proposition}{Propositions}

\newcommand{\E}{\mathbb{E}}

\newcommand{\ceil}[1]{\left\lceil #1 \right\rceil}
\newcommand{\floor}[1]{\left\lfloor #1 \right\rfloor}

\renewcommand{\epsilon}{\varepsilon}

\newcommand{\N}{\mathbb{N}}

\renewcommand{\P}{\mathbb{P}}

\newcommand{\PP}{\mathbb{P}}

\newcommand{\beq}{\begin{equation}}
\newcommand{\eeq}{\end{equation}}

\newcommand{\bX}{\bm{X}}

\usepackage{appendix}
\crefname{appsec}{Appendix}{Appendices}

\begin{document}

\title[Sharp Thresholds Imply Circuit Lower Bounds]{Sharp Thresholds Imply Circuit Lower Bounds:\\ from random 2-SAT to Planted Clique.}

\author[D. Gamarnik, E. Mossel, I. Zadik]{David Gamarnik$^\circ,^\dagger$, Elchanan Mossel$^\star$, and Ilias Zadik$^\ddag,^\star$}
\thanks{\raggedright$^\circ$Institute for Data, Systems, and Society, MIT;
$^\dagger$Operations Research Center \& Sloan School of Management, MIT.
$^\star$Department of Mathematics, MIT. $^\ddag$ Department of Statistics and Data Science, Yale University;\\
Email: \texttt{\{gamarnik,elmos\}@mit.edu, ilias.zadik@yale.edu}}

\begin{abstract}%
We show that sharp thresholds for Boolean functions directly imply average-case circuit lower bounds. More formally we show that any Boolean function exhibiting a sharp enough threshold at an \emph{arbitrary} threshold cannot be computed by Boolean circuits of bounded depth and polynomial size. Our results hold in the average case, in the sense that sharp thresholds imply that for any bounded depth polynomial-size circuit that attempts to compute the function the following holds: for most $p$ near the threshold if the input is drawn from the product measure with bias $p$, the circuit disagrees with the function with constant probability. Interestingly, we also prove a \emph{partial converse result} that if a Boolean function corresponding to a monotone graph property does not exhibit a sharp threshold, then it can be computed by a Boolean circuit of bounded depth and polynomial size on average.

Our general result implies new average-case bounded depth circuit lower bounds in a variety of settings.
\begin{itemize}
    \item[(a)]($k$-cliques) For $k=\Theta(n)$, we prove that any circuit of depth $d$ deciding the presence of a size $k$ clique in a random graph requires exponential-in-$n^{\Theta(1/d)}$ size. To the best of our knowledge, this is the first average-case exponential size lower bound for bounded depth circuits solving the fundamental $k$-clique problem (for any $k=k_n$).
    \item[(b)](random 2-SAT) We prove that any circuit of depth $d$ deciding the satisfiability of a random 2-SAT formula requires exponential-in-$n^{\Theta(1/d)}$ size. To the best of our knowledge, this is the first bounded depth circuit lower bound for random $k$-SAT for any value of $k \geq 2.$ Our results also provide the first rigorous lower bound in agreement with a conjectured, but debated, ``computational hardness'' of random $k$-SAT around its satisfiability threshold. 

    \item[(c)](Statistical estimation -- planted $k$-clique) Over the recent years, multiple statistical estimation problems have also been proven to exhibit a ``statistical'' sharp threshold, called the All-or-Nothing (AoN) phenomenon. We show that AoN also implies circuit lower bounds for statistical problems. As a simple corollary of that, we prove that any circuit of depth $d$ that solves to information-theoretic optimality a ``dense'' variant of the celebrated planted $k$-clique problem requires exponential-in-$n^{\Theta(1/d)}$ size. 
    %To the best of our knowledge, these are the first circuit lower bounds established for statistical estimation tasks.

\end{itemize}

Finally, we mention that our result is in conceptual agreement with a plethora of papers from statistical physics over the last decades suggesting that the presence of a discontinuous phase transition is connected with a certain form of computational hardness.

\end{abstract}

% \begin{keywords}
% \end{keywords}

\maketitle

\date{\today}

\newpage 
%\tableofcontents

\section{Introduction}

Boolean circuits is a formal model of computation which along with Turing machines formalizes the notion of algorithmic decision making. 
Algorithmic complexity is studied in part by analyzing the possible size and depth of a Boolean circuit $\mathcal{C}: \{0,1\}^N \rightarrow \{0,1\}$ (with AND, OR and NOT gates) solving concrete 
decision algorithmic tasks, that is deciding whether the Boolean input vector of the circuit satisfies a property of interest (e.g., the property of a graph containing a clique). This pursuit is greatly motivated by powerful complexity implications of such lower bounds, e.g., if no polynomial-size Boolean circuit can solve some specific NP-hard decision problem then $\mathcal{P} \not = \mathcal{NP}.$ 

Unfortunately, achieving general circuit lower bounds has been proven to be significantly mathematically challenging. For this reason, researchers have focused on establishing lower bounds for circuits that satisfy certain constraints. A notable and very fruitful direction has been the study of monotone circuits (where NOT gates are not allowed); a series of strong size lower bounds have been established for the size of monotone circuits deciding a number of different tasks (e.g., \cite{Razborov-circuit, alon1987monotone}). Similarly, a fruitful direction has been the study of $\mathcal{AC}_0$,
which is the class of bounded depth (i.e, constant depth) polynomial-size circuits. While $\mathcal{AC}_0$ remains less understood than the class of monotone circuits, a series of notable lower bounds for $\mathcal{AC}_0$ has been established, some of which will be discussed below. 

In parallel, a superficially independent line of important work on Boolean functions $f: \{0,1\}^N \rightarrow \{0,1\}$ studied the probability that a $p$-biased random input $\bX \sim \P_p=\mathrm{Bern}(p)^{\otimes N}$  satisfies $f(\bX)=1.$ It has been understood that many Boolean functions exhibit a striking discontinuous phase transition, also known as a \emph{sharp threshold} (see e.g., \cite{Friedgut}): For a threshold $p_c=(p_c)_N \in (0,1)$ and a window size $\epsilon=\epsilon_N \in (0,1), \epsilon_N=o(1)$ it holds 

\[ \lim_N \E_{p}f= \begin{cases}
  0 & \mbox{ if $ p=(1-\epsilon)p_c$}\\
  1 & \mbox{ if $ p=(1+\epsilon)p_c $}
  \end{cases}
  \]Here and everywhere by $\E_p f$ we denote the expectation when $\bX \sim \P_p,$ i.e. $\E_p f=\E_{\bX \sim \P_p} f(\bX).$
  Examples that exhibit such discontinuity span a large set of Boolean functions, including the ones deciding the connectivity of a graph,  the presence of a $k$-clique, and the satisfiability of a $k$-SAT formula.

In this work, we connect the two above lines of work, by establishing a connection between the existence of a sharp threshold for a Boolean function and its circuit complexity. In particular, our main result is the following generic statement.
\begin{theorem}[Informal version of Theorem \ref{thm:main_1}]\label{thm:informal}
    For some small constant $c>0$, any Boolean circuit $\mathcal{C}$ of depth $d \leq c\log N/\log \log N $  which exhibits a sharp threshold with window size $\epsilon=\epsilon_N$ must have size at least $\exp(1/\epsilon_N^{1/d}).$
\end{theorem}Note that our Theorem \ref{thm:informal} applies for sharp thresholds exhibited at any threshold $p_c$. Yet we note that, for this informal presentation of it, we ignored a mild dependency our size lower bound has on the exact location of threshold $p_c$. We direct the reader to Theorem \ref{thm:main_1} for the exact statement.

Our main result is inspired by a discussion of Kalai and Safra in their 20-years old survey \cite{KalaiSafra}, where at a high level they describe an argument why circuits with bounded depth and polynomial size cannot have sharp thresholds. This idea is well motivated by the celebrated Linial-Mansour-Nisan (LMN) Theorem~\cite{linial1993constant}, which states that $\mathcal{AC}_0$ with random inputs have a low degree Fourier expansion, in some appropriate sense. At the same time, another celebrated work of Friedgut~\cite{Friedgut} establishes, conversely, that properties with sharp thresholds cannot be mostly supported on the low degree Fourier terms. This appears already to be giving the desired result stating that $\mathcal{AC}_0$ circuits can't detect properties exhibiting sharp thresholds. The trouble with this argument is that Friedgut's theorem primarily applies to the case when the bias $p=p_N$ is close to zero, while bounds implied by the LMN theorem appear to be  strong enough only when $p$
is constant, and degrade rapidly as $p_N$  converges to zero~\cite{furst1991improved}. Indeed, Kalai and Safra themselves end their high-level argument in \cite{KalaiSafra} by commenting on the LMN-based results needed for their connection, saying that ``\emph{it appears that all these results ... apply ... when $p$ is bounded away from $0$ and $1$.}''

 In our work with Theorem \ref{thm:main_1}  we ``exonerate'' their idea by showing that their suggestion is indeed true, and in fact holds \emph{for all} possible thresholds of the sharp threshold $p_c=(p_c)_N$, even when they tend to $0$ or $1$.
 This is achieved by employing an appropriate modification of a ``de-biasing'' technique on Boolean circuits introduced in~\cite{gamarnik2020low} which works as follows (we direct the reader to e.g., \cite{keller2012simple} for earlier similar results). We first construct for any $p=p_N \in (0,1)$ a polynomial-size ``biasing'' layer $\Phi$ which satisfies for $X \sim \mathrm{Bern}(1/2)^{\otimes N}$, $\Phi(X) \sim \mathrm{Bern}(p)^{\otimes N}.$ Then placing $\Phi$ on top of an arbitrary circuit $\mathcal{C}$ (resulting to a modified circuit $\mathcal{C} \circ \Phi$) we can transform an input of unbiased Bernoulli random variables $X \sim \mathrm{Bern}(1/2)^{\otimes N}$ to $\mathcal{C} \circ \Phi,$ to an input $\Phi(X) \sim \mathrm{Bern}(p)^{\otimes N}$ to the original circuit $\mathcal{C}$. Via this technique, we conclude that if a circuit $\mathcal{C}$ exhibits a sharp threshold at some density (potentially close to zero or one), the modified circuit $\mathcal{C} \circ \Phi$ must exhibit the same sharp threshold but now at constant density (bounded away from zero and one). We then consider any circuit $\mathcal{C}$ with a sharp threshold and appropriately materialize the suggested arguments by Kalai and Safra but importantly on the modified circuit $\mathcal{C} \circ \Phi$ instead.
 This debiasing technique is crucial as it is the key that allows us to obtain new average-case circuit lower bounds for many problems of interest that have threshold close to zero or one, e.g., for deciding the satisfiability of a random 2-SAT formula, where thethreshold is close to zero, and for deciding the existence of a $k=\Theta(n)$-clique in random graphs, where the threshold is close to 1.

 The methodology one can use our main result as a black box to conclude average-case circuit lower bounds is simple and generic. Let us assume some Boolean function $f$ exhibits a sharp threshold of window $\epsilon$ at some threshold $p_c$. Then notice that if some depth $d$ Boolean circuit $\mathcal{C}$ with $d \leq c \log n/\log \log n$ agrees with $f$, with high probability, around the interval $[(1-\epsilon)p_c,(1+\epsilon)p_c]$, then $\mathcal{C}$ must also exhibit a sharp threshold of window $\epsilon$ around $p_c$. In particular, using Theorem \ref{thm:informal}, we can then directly conclude that $\mathcal{C}$ must have size $\exp(\epsilon^{-\Theta(1/d)}).$

As an instructive example, let us consider the case where the Boolean function is the majority function $\mathrm{MAJ}_N$ on $N$ bits. Standard central limit theorem considerations imply that $\mathrm{MAJ}$ has a sharp threshold at $1/2$ with a window size of order $1/\sqrt{N}$. Using Theorem \ref{thm:informal}, we directly conclude (via a new proof) the well-known lower bound, that any depth $d$ Boolean circuit $\mathcal{C}$ with $d \leq c \log N/\log \log N$ agreeing with $\mathrm{MAJ}_N$ on a random input with bias $p \approx 1/2$ must have size $\exp(N^{\Theta(1/d)})$ (see e.g., \cite{Smolensky} for a folklore proof of this fact).

We finally note that in our result we make no attempt to fully optimize the specific constants in front of the depth-dependent term $d$ in our size lower bounds. Moreover, we remind the reader that recent celebrated works (e.g., \cite{rossman2015average}) have been able to obtain optimal depth-size trade-offs by showing that for any constant $d,$ some function computed by a linear-size depth $d$ circuit requires $\exp(n^{\Omega(1/d)})$-size to be computed on average by any $d-1$-depth circuit. We note that the circuit lower bounds derived in this work are not able to capture such a trade-off; if a Boolean function has a sharp threshold then for all constant $d$ all depth $d$ polynomial-size circuits 
 fail to compute the Boolean function on average. In particular, no separation can be implied with polynomial yet we consider the appeal of our result to be that it offer a general and simple framework to derive new circuit lower bounds.

\subsection{Partial converse: monotone graph properties}

Notice that using Theorem \ref{thm:informal} one can directly establish that all Boolean functions exhibiting a sharp threshold with a window size which is scaling super-logarithmically, i.e., $1/\epsilon_n=(\log n)^{\omega(1)}$, cannot be computed, even on average, by a polynomial-size and constant-depth circuit -- more specifically, they cannot be computed by an $\mathcal{AC}_0$ circuit with high probability under the $p$-biased measure when $p$ is in the threshold interval of $f$.  It is natural to wonder whether our condition on $f$ exhibiting (1) a sharp threshold and moreover, (2) having a ``sufficiently small'' window size,  are necessary or of a purely technical nature.

Our second main result is that, perhaps surprisingly, (1) is not a technical condition, and a partial converse holds, as long as the Boolean function $f$ corresponds to a monotone graph property (i.e., it is invariant under graph automorphisms, see e.g., \cite{Friedgut}). We prove that any such $f$ that does not have a sharp threshold, in the sense that the window size satisfies $\epsilon=\Omega(1)$, then it is in fact computable by an $\mathcal{AC}_0$ circuit (in fact, depth four suffices) on average. 

\begin{theorem}[Informal version of Theorem \ref{thm:conv}]\label{thm:informal_conv}
    Any monotone Boolean function $f$ corresponding to a graph property that does not have a sharp threshold, i.e., it's window size satisfies $\epsilon=\Omega(1),$ is computable in $\mathcal{AC}_0$ on average. More specifically, there exists a polynomial-size and 4-depth circuit that agrees with $f$ on average.
\end{theorem}

\begin{remark} Our proof leverages a celebrated characterization of monotone graph properties not exhibiting a sharp threshold by Friedgut \cite{Friedgut}. Importantly, subject to a well-known conjecture from the same work \cite[Conjecture 1.5]{Friedgut}, it directly generalizes to arbitrary monotone Boolean functions that do not exhibit a sharp threshold.
\end{remark}

%\begin{remark}
  %  It is natural to wonder what happens for Boolean functions with a sharp threshold but not significantly small window size, i.e., when $\omega(1)=1/\epsilon_N=(\log n)^{O(1)}.$ An example falling into this category is the Boolean function $f$ corresponding to whether a graph contains or not isolated nodes. Straightforward calculations imply that $f$ corresponds to a monotone property with window size satisfying $1/\epsilon=\Theta(\log n).$ Moreover, it is easy to construct an $\mathcal{AC}_0$ circuit that detects (even worst-case) whether a graph has an isolated node or not. \textcolor{red}{TODO: To be continued}
%\end{remark}
%%NO optimal vs Hastad

\subsection{Circuit lower bound for $k=\Theta(n)$-clique in Erdos-Renyi graphs}

We now formally start the applications of our main result with the $k$-clique problem, a very well-studied $\mathcal{NP}$-hard problem. This setting corresponds to a Boolean function $f: \{0,1\}^{\binom{n}{2}} \rightarrow \{0,1\}$ where the input is the adjacency matrix of an undirected graph on $n$ vertices, and output is $1$ if and only if the graph has a clique of size $k$.

The circuit complexity of this problem has been extensively studied. Strong (exponential in $n^{\Theta(1)}$) lower bounds on the size of a circuit for $k$-clique exist for some appropriate value of $k=k_n$,
when the circuit is restricted to be monotone, namely not containing negation gates \cite{Razborov-circuit, alon1987monotone, arora2009computational}. However, exponential lower bounds for \emph{not necessarily} monotone Boolean circuits appear harder to achieve. As explained above, a natural step in this direction is to consider (potentially non-monotone) circuits but of bounded/controlled depth. One of the earliest such results was by Lynch~\cite{lynch1986depth},
who showed that when the cardinality $k$ of the clique is \emph{assumed to be} most $\log n$, the depth-$d$  circuit which detects the presence of such cliques 
must have size at least $n^{\Omega\left(\sqrt{k/d^3}\right)}$. 
The current state of the art is a celebrated result by Rossman~\cite{rossman2008constant} and applies to the case
when $k=o(\sqrt{\log n})$. He  established a depth-independent average-case lower bound 
$\omega\left(n^{2(k+1)/9}\right)$ when the circuit depth $d$ is constant.

Using our Theorem~\ref{thm:informal} we are able to obtain the first to our knowledge average-case bounded depth circuit lower bound for the $k$-clique problem when $k=\Theta(n)$. 

\begin{theorem}[Informal, $k$-clique (see Theorem \ref{thm:clique})]\label{thm:informal_clique}
    Suppose $k=\Theta(n)$. Any Boolean circuit $\mathcal{C}$ of depth $d \leq  \log n/(2\log \log n)$ which solves the $k=\Theta(n)$-clique problem on any Erdos-Renyi graph $G(n,p)$, must have size at least $\exp(n^{\frac{1}{2 d}})$.
    In particular, for bounded depth circuits, their size needs to be exponential in $n^{\Theta(1)}.$
\end{theorem}

 We highlight that our lower bound holds as long as the depth of the circuit scales less than $\log n/(2\log\log n)$, while all previous circuit lower bounds for $k$-clique we are aware of require the depth to be $o(\log n/\log \log n)$ ~\cite{rossman2008constant} to be applicable.

We also note that in light of Theorem \ref{thm:informal}, the proof of Theorem~\ref{thm:informal_clique} follows from a rather straightforward argument. Based on classical results on the independence number of random graphs \cite{frieze1990independence} and standard concentration inequalities, we prove that  $k=\Theta(n)$-clique Boolean function exhibits a sharp threshold with window size $1/n^{1/2}.$ From this, using Theorem \ref{thm:informal} the result readily follows. \footnote{Formally, this argument goes through as described for the ``complement'' of the $k$-clique Boolean function which decides if a graph has an independent set of size $k$. We omit these details, and direct the reader to Section \ref{sec:apps_cliques} for the exact statement and proof.} 

Finally it is important to mention that similar lower bounds for the $k$-clique problem in random graphs have been recently established in~\cite{gamarnik2020low}, but importantly for the \emph{search version} of the problem. Namely, the search problem asks for a circuit which outputs the position of the vertices of a $k$-clique (if it exists). We remind the reader that the (weaker) decision problem we study asks to solely decide the presence of a $k$-clique in the graph, and not necessarily to output the location of one. The search aspect is an important limitation of the result in~\cite{gamarnik2020low} as the search problem is naturally harder to solve, making the proof of a lower bound easier.

\subsection{Circuit lower bound for random 2-SAT}

We now turn to the fundamental problem of $k$-SAT. The $k$-SAT decision setting can be modeled using a Boolean function $f: \{0,1\}^{m} \rightarrow \{0,1\}$ with input a formula which is the conjunction of $m$ clauses, each of which is a disjunction of at most $k$ out of $n$ literals. The Boolean function is 1 if and only the formula is satisfiable. The $k$-SAT decision problem asks to decide if a given $k$-SAT formula is satisfiable or not. It is known that the decision problem for $2$-SAT is in $\mathcal{P}$, while the decision problem for $k$-SAT for $k \geq 3$ is $\mathcal{NP}$-hard \cite{Cook71}.

Random $k$-SAT naturally corresponds to the case where the $k$-SAT formula consists of $m$ clauses chosen uniformly at random. Over the recent years, a series of celebrated works have identified the ``satisfiability'' sharp threshold for random $k$-SAT (for certain constant values of $k$) that is the critical ratio of number of clauses to number of variables $m/n$ for which random $k$-SAT becomes unsatisfiable from satisfiable. More specifically, we direct the reader to \cite{goerdt1996threshold, Chvatal92} for the case $k=2$ and \cite{ding2015proof} for a celebrated proof when $k$ is a large constant, alongside with many references (see also \cite{Friedgut} for a slightly weaker notion of a sharp threshold but for all $k$). Interestingly, when $k=2$ the seminal work of \cite{bollobas2001scaling} established also the window of the sharp threshold to be of order $\epsilon=\Theta(n^{-1/3})$. To the best of our knowledge, the window size when $k\geq 3$ remains unknown.

It is a very natural question whether computationally efficient algorithms can solve the decision problem on random $k$-SAT. Clearly there is a polynomial-time when $k=2$, but even if $k \geq 3$, in case we are promised to be ``away'' from the sharp threshold, there is always a trivial polynomial-time algorithm that succeeds with high probability: if the ratio $m/n$ is above the sharp satisfiability threshold the algorithm should always say ``UNSAT'', and if we are below the algorithm should always say ``SAT''. This trivial algorithm suggests that if there is a ``hard'' region for random $k$-SAT it should be \emph{around the sharp threshold.}

We are not aware of any rigorous result for the complexity of random $k$-SAT around the threshold, but some almost 30-years old simulations are suggestive of a non-trivial landscape. By employing several SAT solvers authors have predicted that the decision problem of random k-SAT, for various small values $k$, will be significantly computationally ``harder'' around the threshold, and significantly computationally ``easier'' away from it (see \cite{selman1996generating, selman1996critical, mitchell1996some} and in particular \cite[Figure 2]{mitchell1996some} for $k=2,3,4,5.$) It is worth noting that the increase in computational difficulty at the threshold applies also to the case where $k=2$. The authors of \cite{selman1996generating} explain this predicted ``easy to hard(er) to easy'' computational transition by arguing that for $m$ much smaller than the critical threshold there are many satisfiable assignments, and for $m$ much larger it is easy to find a contradiction to satisfiability. In recent works, it appears that some authors have though perceived these findings to be dependent on the choice of SAT solver, and whether there is a hard(er) phase around the threshold or not, has been a topic of significant interest from both computer science and statistical physics communities  (see \cite{monasson1999determining, coarfa2000random} and references therein). From a mathematical standpoint, the following question is naturally emergent from the aforementioned works in the 90's and remarkably open:
\begin{quote}
    \centering{\emph{Is there a rigorous computational complexity lower bound for the decision problem of random $k$-SAT around the satisfiability threshold?}}
\end{quote}

 Using Theorem \ref{thm:informal} we provide a partial answer to the question above by studying the case $k=2$ and considering the class $\mathcal{AC}_0$ of polynomial-size bounded depth circuits. We prove that $\mathcal{AC}_0$ cannot solve random $2$-SAT around the satisfiability threshold. Our result follows straightforwardly from the $\epsilon=\Theta(n^{-1/3})$ small-window of the established random $2$-SAT phase transition \cite{bollobas2001scaling}. We prove the following result.

\begin{theorem}[Informal, $2$-SAT (see Theorem \ref{thm:sat})]\label{thm:informal_SAT}
    Any Boolean circuit $\mathcal{C}$ of depth \\$d \leq  \log n/(3\log \log n)$ which decides the satisfiability of a random $2$-SAT formula, must have size at least $\exp(n^{\frac{1}{3 d}}).$ In particular, $\mathcal{AC}_0$ cannot solve random 2-SAT around the threshold.
\end{theorem}

While our result does not of course answer the question for the class of polynomial-time algorithms, it is an unconditional lower bound that applies even for random $2$-SAT (which is trivially in $\mathcal{P}$). Naturally, if random $k$-SAT for $k \geq 3$ is proven to exhibit an inverse polynomially small window, our lower bound would generalize in this case as well.

\subsection{Circuit lower bounds from statistical sharp thresholds}

We now make a twist in our applications, and turn our focus from decision settings of random structures to statistical estimation tasks. Our motivation is that over the recent years a notable amount of work has studied the so-called ``computational-statistical trade-offs'', that is computational challenging regimes appearing in a plethora of high dimensional statistical estimation tasks. Such settings include the celebrated planted clique problem \cite{jerrum1992large}, sparse regression \cite{gamarnik2022sparse}, tensor PCA \cite{richard2014statistical}, and more. Interestingly, a (worst-case) complexity-theory explanation of most of these trade-offs remains eluding. For this reason, a significant line of work attempts to explain these trade-offs by proving restricted lower bounds against subclasses of polynomial-time estimators. These have so-far included (but are not limited to) the class of low-degree polynomial estimators \cite{schramm2022computational}, estimators using statistical queries \cite{feldman2017statistical}, MCMC methods \cite{chen2023almost} and Langevin dynamics \cite{arous2018algorithmic}.

The connection with the present work lies on an, again superficially disjoint, series of recent works in high dimensional statistics which has proved that many statistical estimation settings exhibit a sharp threshold in terms now of the possible estimation/recovery of a hidden signal, known as the All-or-Nothing (AoN) phenomenon \cite{gamarnik2022sparse, reeves2021all}. More specifically, the AoN phenomenon appears when at some critical noise level $\sigma=\sigma_N>0$ (also known as the information-theoretic threshold and critical window size $\epsilon=\epsilon_N \in (0,1)$ the following holds: for noise level $(1-\epsilon)\sigma$ almost perfect estimation of the hidden signal is possible (the ``all'' phase), while for noise level $(1+\epsilon)\sigma$ even weak estimation of the hidden signal is not possible (the ``nothing'' phase). Interestingly, AoN has been proven to appear in many models also exhibiting computational-statistical trade-offs, such as planted clique \cite{mossel2023sharp}, sparse tensor PCA \cite{niles2020all} and sparse regression \cite{reeves2021all}. For this reason, the recent work \cite{mossel2023sharp} asked whether this is a mere coincidence or one can obtain a result connecting the presence of AoN with computational challenges.

In this work, we provide an answer to the question of \cite{mossel2023sharp} and show that our main result Theorem \ref{thm:informal} can also be adapted to obtain \emph{circuit lower bounds directly from the AoN phenomenon}. We prove that (a slight variant of) the existence of AoN implies circuit lower bounds for a large class of Boolean high dimensional statistical estimation problems called ``Hidden Subset Problems'' (see Definition \ref{dfn:hidden}) introduced in \cite{mossel2023sharp}. 

\begin{theorem}[Informal version of Theorem \ref{thm:circuit_aon}]\label{thm:informal_aon}
For some small $c>0$ the following holds. For a large class of Boolean statistical estimation settings, if they exhibit the All-or-Nothing phenomenon with window size $\epsilon=\epsilon_N>0,$ then any circuit $\mathcal{C}$ of depth $d \leq c \log N/\log \log N$ correctly estimating the hidden signal to information-theoretic optimality must have size at least $\exp(1/\epsilon_N^{1/d}).$
\end{theorem}

The proof proceeds at a high level as follows. Assume the existence of an estimator, computed by a bounded depth circuit, which is successful all the way to the information-theoretic noise level for the task of interest. By definition of the information-theoretic noise level, the same circuit needs to fail for values a bit above this critical noise level (as all estimators fail). Hence, this circuit is exhibiting a certain ``statistical'' sharp threshold. With a little bit of work, we modify the circuit to exhibit the classic notion of a sharp threshold from Theorem \ref{thm:informal}, from which we conclude the result.

We note previous results on the average-case complexity of inferring the root label in the broadcast problem~\cite{moitra2020parallels}. The results of~\cite{moitra2020parallels} provide circuit lower bounds for a natural and well-studied inference problem of inferring the root in the broadcast model. Similar to the 2-SAT problem, there is a linear time algorithm for solving this task as well. In this setting the results of~\cite{moitra2020parallels} prove circuit lower bounds for this problem. Compared to our results, the results of~\cite{moitra2020parallels} provide less explicit bounds on the required size in terms of depth for $\mathcal{AC}_0$ circuit. On the other hand, some of the results of~\cite{moitra2020parallels} provide average-case $\mathcal{NC}_1$ hardness results. Since $\mathcal{TC}_0$ contains the majority function (which exhibits a sharp threshold) and $\mathcal{TC}_0$ is contained in $\mathcal{NC}_1$, it seems like it is impossible to prove $\mathcal{NC}_1$ hardness results based on the sharp threshold technique of the present work. 

It is perhaps natural to ask whether the class of bounded depth polynomial-size circuits is a reasonable class of estimators for statistical problems. Indeed, given the folklore result that a bounded-depth polynomial-size circuits cannot compute the Majority function, multiple common statistical estimators (e.g. based on thresholding empirical moments) can also be proven to not be computable by bounded-depth circuits. We note though that one may not wish to compute the Majority function because of \emph{space constraints} -- the Majority needs $\log n$ bits of memory to be computed. In fact, motivated by such concerns, a large number of results have studied the complexity of testing the bias of a single coin by focusing only on classes of estimators with restricted memory (that do not contain Majority). For example, see  \cite{amano2009bounds} for $\mathcal{AC}_0$ upper bounds, \cite{aaronson2010bqp} for $\mathcal{AC}_0$ lower bounds, and \cite{brody2010coin} for lower bounds against read-one branching programs. Interestingly, the cited $\mathcal{AC}_0$ lower bound for the bias problem is also based on an argument using the sharp threshold behavior of Majority. Now, from this point of view, Theorem \ref{thm:informal_aon} offers a general method for how one can obtain $\mathcal{AC}_0$ lower bounds for a variety of high dimensional statistical settings based on statistical sharp thresholds. Such settings imply multiple new circuit lower bounds well beyond the single coin estimation setting 
 (e.g., see the next section for the planted clique model.)

Before proceeding with applications, we note a different motivation for Theorem \ref{thm:informal_aon} from the study of computational-statistical trade-offs. The ``standard'' restricted lower bounds for computational-statistical trade-offs focus (with few exceptions) on the following classes: low-degree polynomials, convex programs, MCMC and local-search methods, SQ methods, and message passing procedures. These are very powerful and statistically natural classes to consider, but during the recent years it has become increasingly clear that the class of polynomial-time estimators can be significantly stronger than the above classes (e.g. lattice-based methods from cryptography have been proven to surpass all of these lower bounds in clustering mixtures of Gaussians \cite{zadik2022lattice}). On the other hand, for Boolean problems, the framework of polynomial-size circuits is known to capture all polynomial-time algorithms. It is interesting to wonder whether studying estimators as Boolean circuits can help us improve our understanding of the computational complexity of statistical problems. For example, it is well-known in circuit complexity that the class we consider in this work, of bounded depth polynomial-size circuits, includes the logarithmic time hierarchy (see e.g., \cite{vollmer1999introduction}). This makes it therefore a class of polynomial-time estimators that is not captured from the above-mentioned well-studied classes.  Our Theorem  \ref{thm:informal_aon} offers a generic method about how one can argue the failure of bounded depth circuits using the AoN phenomenon from the high dimensional statistics literature. Our hope is that our result will generate more works using the technology of circuit lower bounds but adapted to statistical tasks, hopefully shedding more light on computational-statistical trade-offs and connecting it more with more traditional computational complexity theory techniques.

\subsection{Circuit lower bounds for Planted Clique}
As a concrete application of Theorem \ref{thm:informal_aon}, we prove circuit lower bounds for a ``dense'' variant of the celebrated planted $k$-clique problem.

In this setting, a statistician observes an $n$-vertex undirected graph which is the union of a (randomly placed) $k=k_n$-clique $\mathcal{PC}$ with a sample from the Erdos-Renyi graph $G(n,p).$ Here $1/p$ corresponds to the noise level. The goal of the statistician is to exactly recover the vertices of $\mathcal{PC}$, given access to the graph, with probability at least $0.99.$ We explore a strong notion of algorithmic success: we fix $k$ and say that an algorithm \emph{solves (to information-theoretic optimality) the $k$-planted clique problem}, if it succeeds for any inverse noise level $p=p_n>0$ for which it is information-theoretically possible. This is a non-trivial notion of algorithmic success as for example the optimal brute force search for a $k$-clique estimator always enjoys such guarantees.

Observe that our perspective is that we fix $k=k_n$ and vary $p=p_n$. The planted clique model has been extensively studied in the slightly dual setting where one fixes the inverse noise level $p$ (mostly setting $p=1/2$) and varies $k=k_n$. The reason we flip these roles is to be able to take advantage of the AoN literature where varying the noise level is more natural. Nevertheless, in the case $p=1/2$ it is a folklore result that the information-theory threshold for recovery is $k=2\log_2 n$ and in particular if  $k >2 \log_2 n$ recovery of $\mathcal{PC}$ is possible. Interestingly though all known polynomial-time algorithms require $k=\Omega(\sqrt{n})$, i.e., an exponentially larger clique to succeed \cite{alon1998finding}.  Multiple restricted lower bounds have been established in the regime $2 \log_2 n <k=o(\sqrt{n})$, including failure of low-degree methods \cite{schramm2022computational}, of MCMC methods \cite{jerrum1992large, chen2023almost} and more, to support the computational-statistical trade-off of the model could be of fundamental nature.

Now, back in our setting where we fix $k$ and vary $p$, Theorem \ref{thm:informal_aon} is able to refute bounded depth circuits from achieving the information-theoretic threshold $p_{\mathrm{IT}}$. Hence, to consider the case where $p=1/2$ we should choose $k=2\log_2 n$ which is the unique value of $k=k_n$ with $p_{\mathrm{IT}}=1/2$. In this case planted clique indeed does exhibit AoN \cite{mossel2023sharp}, but the window size can be easily derived from existing work to be $(\mathrm{poly}(\log n))^{-1}$. This small window size unfortunately leads to a degenerate size lower bound for bounded-depth circuits based on Theorem \ref{thm:informal_aon}. 

For this reason, we consider instead a ``denser'' case where we plant a $k=(\log n)^{\omega(1)}$ clique in $G(n,p)$ where now the information-theory threshold is $p_{\mathrm{IT}}=1-\Theta(\log n/k),$ and in particular converges to one (hence the ``denser'' name). In that case, we prove that AoN appears with a window size $\epsilon=\Theta(1/k)$, which now is large enough for us to conclude non-trivial circuit lower bounds.

We prove the following.
\begin{theorem}[Informal version of Theorem \ref{thm:planted-clique}]\label{thm:informal_aon_pc}
For some small $c>0$ the following holds. Assume $(\log n)^{\omega(1)}=k \leq n^{1/3-\epsilon}$. Then, any Boolean circuit $\mathcal{C}$ of depth $d \leq  c\log n/\log \log n$ solving the planted $k$-clique problem to information-theoretic optimality, must have size at least $\exp(k^{1/d}).$
\end{theorem}The proof of the requires AoN result follows from an application of the so-called planting trick, initially introduced in the context of random constraint satisfaction problems \cite{achlioptas2008algorithmic} alongside a careful application of the second-moment method.

\subsection{Notation}\label{sec:notation}
$\P_p$ corresponds to the product measure $\mathrm{Bernoulli}(p)^{\otimes N}$ and $\E_p$ the expectation with respect to $\P_p$. For $x \in \{0,1\}^N$ we denote by $x^{-i} \in \{0,1\}^{N-1},$ the vector where $i$-th bit is deleted, and by $(x^{-i},1) \in \{0,1\}^N$ (respectively $(x^{-i},0)\in \{0,1\}^N$) the modified binary vector $x$ where the $i$-bit is set to be $1$ (respectively $0$).  We define the influence of the bit $x_i$ to $f$ by $(I_i)_p(f):=p(1-p)\E_p (f(x^{-i},1)-f(x^{-i},0))^2,$ and the total influence $I_p(f)=\sum_{i=1}^N (I_i)_p(f).$
 All logarithms  are  base 2. With high probability (w.h.p.) means with probability going to one as $N$ grows. We consider circuits with AND, OR, and NOT gates of unbounded fan-in. As customary, we assume that the NOT gates are only at the input level, which is an assumption that does not increase the depth of a circuit and at most doubles the size of the circuit.

\section{Main Result}

We start with defining the sharp threshold property of a Boolean function.
\begin{definition}\label{dfn:sharp}
     We say $f: \{0,1\}^N \rightarrow \{0,1\}$ exhibits a \emph{sharp threshold} with window size $\epsilon=\epsilon_N \in (0,1)$ if for some threshold $0<p_{c}=(p_c)_N\in (0,1)$ and jump size $\delta=\delta_N \in (0,1)$ with $\delta/\epsilon=\omega(1)$ the following holds for large enough $N$:
 \begin{align*}
       | \E_{(1+\epsilon)p_c}f -      \E_{(1-\epsilon)p_c}f | \geq \delta.
    \end{align*}
\end{definition}

Note that this definition generalizes the one we gave in the introduction as the ``jump size'' $\delta$ does not need to be close to one, and in fact can even shrink to zero, i.e., $\delta=\delta_N=o(1).$

Our main theorem is the following.

\begin{theorem}\label{thm:main_1}
For some universal constants $c_1,c_2>0$ the following holds.
    Let $C: \{0,1\}^N \rightarrow \{0,1\}$ be a circuit of size $s=s_N$ and depth $d=d_N.$ Suppose that $\mathcal{C}$ exhibits a sharp threshold with window size $\epsilon=\epsilon_N$ and jump size $\delta=\delta_N$ at the threshold $p_c=(p_c)_N.$

    Let $\beta:=\min\{p_c,1-p_c\}$ and assume $\beta=N^{-\Omega(1)}$ and $\epsilon=o((1-p_c)/\log 1/\beta)$. Then for some universal constants $c_1,c_2,c_3>0,$ the following holds for large enough values of $N$.
\begin{itemize}
    \item[(1)] Either,

    \begin{align}\label{eq:final_0000}
      d \geq \frac{1}{2\log \log N }\log \left[\frac{\delta(1-p_c)}{\epsilon \log 1/\beta}\right] -3,
    \end{align}i.e., the ``depth is large''.

    \item[(2)] Or
    \begin{align}\label{eq:final_001}
     s \geq c_1\exp \left( c_2\left(\frac{\delta(1-p_c)}{\epsilon \log 1/\beta} \right)^{1/(d+3)}\right),
    \end{align}i.e., the ``size is large''.
\end{itemize}

\end{theorem}

\begin{remark}\label{rem:quantity}
The crucial quantity controlling the depth and size lower bounds in Theorem \ref{thm:main_1} is \begin{align}\label{eq:key_term}
\frac{\delta(1-p_c)}{\epsilon \log 1/\beta}=\frac{\delta}{\epsilon} \times \frac{1-p_c}{\log 1/\beta}.
\end{align}The first term $\delta/\epsilon$ is rather intuitive and reads as the \emph{ratio of the jump size to the window size}. One can interpret this naturally as a measure of the ``sharpness'' of the threshold; naturally, the sharper the transition is, the stronger becomes our lower bound. 

The second term $(1-p_c)/\log 1/\beta,$ vanishes when $p_c$ is sufficiently close to zero or one. While the fact that this term vanishes does not affect any of the applications we consider in Section \ref{sec:apps} (where in all cases, $p_c$ tends to either zero or one) it is possible that this term appears solely for technical reasons.

\end{remark}

\begin{remark}
   The Theorem \ref{thm:main_1} requires two conditions for the sharp threshold, that $\beta=N^{-\Omega(1)}$ and $\epsilon=o((1-p_c)/\log 1/\beta)$.

   The first condition $\beta=N^{-\Omega(1)}$ is purely done to ease a bit notation in \eqref{eq:final_0000} and \eqref{eq:final_001}, but it is very natural and we are not aware of any ``natural'' sharp threshold that violates it. In fact, if the condition does not hold, a sharp threshold is impossible unless the jump size is super polynomially small, i.e. $\delta=N^{-\omega(1)}.$ To see this, notice that if $\beta=N^{-\Omega(1)}$ is violated, then, without loss of generality we have $p_c=N^{-\omega(1)}.$Therefore,
\begin{align*}
        \P_{(1+\epsilon)p_c}(\bX =0) = \P_{(1-\epsilon)p_c}(\bX =0)=1-N^{-\omega(1)},
    \end{align*}and hence for any Boolean $f,$
    \begin{align*}
        |\E_{(1+\epsilon)p_c}f - \E_{(1-\epsilon)p_c}f|=|\P_{(1+\epsilon)p_c}(f(0)=1) -\P_{(1-\epsilon)p_c}(f(0)=1) +N^{-\omega(1)}|=N^{-\omega(1)}.
    \end{align*}

   For the first condition, we first highlight that it is in fact necessarily true for our bounds to be meaningful; e.g., the size lower bound \eqref{eq:final_001} trivializes otherwise. That being said, we are unsure if the condition is of fundamental importance and it could be appearing solely of technical nature.

\end{remark}

\begin{remark}\label{rem:assym}

Note that there is a particular type of asymmetry in our circuit lower bounds, and in particular in \eqref{eq:key_term} with respect to the point $p_c=1/2$. That is, if the threshold $p_c$ is closer to one our depth and size lower bounds appear to be getting worse by a factor $(1-p_c)$.

We highlight that this asymmetry is superficial. To see this, define for any circuit $\mathcal{C},$ the circuit $\mathcal{C}'(X):=\mathcal{C}(\neg X)$. Notice $\mathcal{C}'$ has the same depth and size with $\mathcal{C}$ and satisfies for all $p \in (0,1)$ by the obvious coupling, 
$$ \E_p\mathcal{C}'=\E_{1-p}\mathcal{C}.$$Hence, when a circuit $\mathcal{C}$ exhibits a sharp threshold at $p_c$ with window $\epsilon,$ then the circuit $\mathcal{C}'$ exhibits a threshold around now $1-p_c$ with window $\epsilon'=\epsilon p_c/(1-p_c)$\footnote{This follows since $1-(1+\epsilon)p_c=(1-\epsilon')(1-p_c)$ and $1-(1-\epsilon)p_c=(1+\epsilon')(1-p_c).$}. In particular, any threshold for $\mathcal{C}$ around $p_c=1-o(1)$ of window $\epsilon$ corresponds to a sharp threshold for $\mathcal{C}'$ of the same depth and size at $1-p_c=o(1)$ of a \emph{larger window size} $\Theta(\epsilon/(1-p_c))$. Yet, naturally, this should be accounted as the same sharp threshold, just appearing in the ``equivalent'' $\mathcal{C}$ and $\mathcal{C'}.$ In particular, a ``good'' circuit lower bound should be implying the same depth and size lower bounds on both the ``equivalent'' $\mathcal{C}$ and $\mathcal{C'}$. This is exactly the source of asymmetry in our bounds in Theorem \ref{thm:main_1} where the window of the transition is corrected by $(1-p_c)$ when $p_c$ tends to one.

\end{remark}

\begin{remark}
We finally remark that recent celebrated works (e.g., \cite{rossman2015average}) have been able to obtain optimal depth-size trade-offs for bounded depth circuits. In more detail, they have shown that for any constant $d,$ there exists some function computed by a polynomial-size depth $d$ circuit which requires $\exp(n^{\Omega(1/d)})$-size to be computed on average by any $\leq d-1$-depth circuit. We note that while our approach implies a number of circuit lower bounds as direct implications, none of these lower bounds would be able to capture such a strong depth separation.

Indeed, if a Boolean function has a sharp threshold with parameters $\epsilon, \delta, p_c$ then there are two cases. Either the ``crucial'' term $\delta (1-p_c)/[\epsilon \log (1/\beta)]$ in \eqref{eq:key_term} is $\exp(\Theta(N))$ for infinitely many $N$ or $\exp(o(N))$ for infinitely many $N$. In the first case Theorem \ref{thm:main_1} implies a $\exp(N^{\Theta(1/d)})$ size lower bound for all constant $d$. In the latter case, by checking \eqref{eq:final_001} we conclude no $\exp(N^{\Theta(1/d)})$ size lower bound can be implied by Theorem \ref{thm:main_1} for any constant $d$. 
\end{remark}

\begin{remark}\label{rem:log}
Recall that $\mathcal{AC}_0$ is the class of $d=O(1)$-depth and polynomial size circuits. A direct application of Theorem \ref{thm:main_1} is that for a sufficiently large constant $C>0,$ no $\mathcal{AC}_0$ circuit exhibits a sharp threshold with $\delta=\Omega(1)$ and $1/\epsilon \geq (\log N)^{d+C}$.

We now argue that the required poly-logarithmic window for any $\mathcal{AC}_0$ circuit -- i.e., that whenever they exhibit a sharp threshold, they must have window size $\epsilon \geq 1/(\log N)^{d+C}$ --  is tight for any $2 \leq d=O(1)$ up to multiplicative constants in the exponent of $\log N$. In particular, there are $\mathcal{AC}_0$ circuits of depth $d$ exhibiting a sharp threshold with a window size  $\epsilon=O(1/\log^{d/2} N)$.

Let us start with the case $d=2$. Recall that the tribes function on $N$ variables is the OR of $N/w$ disjoint ANDs on $w$ variables each, where we choose $w = \log_2 N - \log_2 \ln N + O(1)$. This choice is so that, under the unbiased measure, the probability that each AND is $1$ is $2 \ln N /N (1+o(1))$ and hence the overall expected value of the function equals $(1-2 \ln N /N(1+o(1)))^{N/w}=\Theta(1)$. Notice that clearly the tribes function is in $\mathcal{AC}_0$ with a circuit of depth $2.$ Also, if for some $\delta>0$ we use the measure with bias $1/2 \pm \delta$ for $|\delta| = O(1/ \log N)$ the probability that each AND clause in the tribes function is $1$ is now $ e^{2  \delta \ln N(1+o(1))} \ln N /N $ and the overall expected value of the function equals $(1-e^{2 \delta \ln N (1+o(1))} 2 \ln N /N)^{N/w}=e^{\Theta(\delta \ln N )}$. As a direct outcome, the window size of tribes is of length $\Theta(1 / \log N)$.

One can generalize the above construction to any even $d \geq 2$. For $d=4$, let $f$ be a tribe function on $\floor{\sqrt{N}}$ bits then we can define $g$ to be the tribes on the $\floor{\sqrt{N}}$ bits obtained by applying on each of the $\floor{\sqrt{N}}$ consecutive blocks of $\Theta(\sqrt{N})$ size, i.e.,
\[
g(x_1,\ldots,x_N) = f(f(x_1,\ldots,x_{\floor{\sqrt{N}}}),\ldots,f(...,x_N))
\]
Then, a direct calculation similar to the one above, implies that the window size for $g$ is of size of order $1/\log^2 n$.
Similarly, for arbitrary even $d \in \mathbb{N}, d \geq 2$ by using the tribes function of size $\floor{N^{2/d}}$
and decomposing $d/2$ times we get a function that is an $\mathcal{AC}_0$ circuit of depth $ d$ and of size $O(N)$ with window size is of order $\Theta(1/\log^{d/2} N)$. 

%It is natural to ask if this is tight. In other words, is it true that a (monotone) AC0 circuit of poly size and depth $d$ has a threshold interval whose length is at least 
%$\log n^{-c d}$ for some $c > 0$.
\end{remark}
\subsection{Converse: monotone graph properties}

Here we prove a partial converse to our main result Theorem \ref{thm:main_1}. We first introduce some notation, specifically used only at this section. 

Let $f: \{0,1\}^N \rightarrow \{0,1\}$ be a Boolean function. We assume $N=\binom{n}{2}$ for some $n \geq 1$ and in particular this allows us to interpret each input $X \in \{0,1\}^N$ as the indicator vector of the edges of an undirected graph on $n$ vertices. We say $f$ corresponds to a monotone graph property if the following two conditions hold: \begin{itemize}
    \item[(1)] (Monotonicity) It holds $f(X) \leq f(Y)$ whenever $X_i \leq Y_i$ for $i=1,\ldots,N.$ 
    \item[(2)] (Invariance) For any automorphism $\phi$ of (the graph) $X$ it holds $f(\phi(X))=f(X).$
\end{itemize}

Now, for any $\alpha \in (0,1)$, we define $p_{\alpha} \in (0,1)$ to be the unique value such that $\E_{p_{\alpha}}f=\alpha.$ A folklore result by Bollobas-Thomason \cite{bollobas1987threshold} implies that for any fixed (independent of $N$) $\alpha \in (0,1)$ and any monotone Boolean function $f$ it holds $p_{\alpha}=\Theta(p_{1/2}).$ 

 We will need an equivalent way to define a sharp threshold compared to Definition \ref{dfn:sharp}. Notice that Definition \ref{dfn:sharp} implies by the mean value theorem that for some $p=p_N \in (0,1)$ it holds
\begin{align}\label{eq:deriv_motiv} p \frac{d}{dp} \E_p f =\Omega(\delta/\epsilon)=\omega(1).
    \end{align} Moreover, by simply using the definition of the derivative, \eqref{eq:deriv_motiv} is also implying the existence of sharp threshold per Definition \ref{dfn:sharp}, making \eqref{eq:deriv_motiv} equivalent with Definition \ref{dfn:sharp}.

    %In that language in all applications of Theorem \ref{thm:main_1}, our approach is to prove that $f$ satisfies an appropriate version of \eqref{eq:deriv_motiv} for $p=p_{1/2}$ and then prove that if a bounded depth circuit could compute $f$ based on Definition \ref{dfn:approx} it must also have a sharp threshold. In particular, by Theorem \ref{thm:main_1} this successful circuit must have either large size or large depth.

We prove the following partial converse to Theorem \ref{thm:main_1}. Suppose that for some arbitrary $\alpha \in (0,1/2)$, \eqref{eq:deriv_motiv} does not hold for any $p \in [p_{\alpha},p_{1-\alpha}]$ (and hence $f$ does not exhibit a sharp threshold at $[p_{\alpha},p_{1-\alpha}]$). We then prove that the Boolean function $f$ can in fact be computed by an $\mathcal{AC}_0$ circuit with high probability over the $p$-biased measure in a uniform way over all  $p \in [p_{\alpha},p_{1-\alpha}].$

\begin{theorem}\label{thm:conv}
    Assume a monotone graph property $f: \{0,1\}^N \rightarrow \{0,1\}, N=\binom{n}{2}$ with $\omega(1/n^2)=p_c=1-\Omega(1)$ such that for some constants $\alpha \in (0,1)$ and $C>0,$  \begin{align}\label{eq:deriv}\sup_{p \in [p_{\alpha},p_{1-\alpha}]} p \frac{d}{dp} \E_p f \leq C.
    \end{align}
    Then for any $\gamma>0,$ there exists a constant $k(\gamma,\alpha,C)>0$ and  an $\mathcal{AC}_0$ circuit $\mathcal{C}: \{0,1\}^N \rightarrow \{0,1\}$ of depth four and size at most $N^{k(\gamma,\alpha,C)}$, such that 
    $$\sup_{p \in [p_{\alpha},p_{1-\alpha}]} | \E_p f-\E_p \mathcal{C}| \leq \gamma.$$
\end{theorem}

\begin{remark}
 Consider any Boolean function $f$ corresponding to a monotone graph property. Theorem \ref{thm:main_1} implies that if $f$ exhibits a sharp threshold with jump size $\delta=\Theta(1)$ and window size $\epsilon$ satisfying $1/\epsilon=(\log N)^{\omega(1)}$ then $f$ cannot be computed in $\mathcal{AC}_0$ on average, while Theorem \ref{thm:conv} implies that if for any jump size of $f$, say $\delta=\Omega(1)$, the minimal window size satisfies $\epsilon=\Omega(1)$ (i.e., $f$ does not exhibit a sharp threshold) then $f$ can be computed in $\mathcal{AC}_0$ on average.

    Notice that the above characterization leaves open the case where $f$ exhibits a sharp threshold with jump size $\delta=\Theta(1),$ but with window size satisfying $\omega(1)=1/\epsilon=(\log N)^{O(1)}$.  It is natural to wonder whether one of the two directions can be extended to this case. Based on Remark \ref{rem:log} we know that using the tribes boolean function, there are $\mathcal{AC}_0$ circuits with sharp thresholds of such order of window sizes. Hence, the only possibility would be if our positive result, Theorem \ref{thm:conv}, could be extended to such window sizes. 

    We are unsure about the answer to this interesting question, but we remark that a standard first and second moment method implies that for $k=O(\log n)$ the Boolean function $f$ detecting the existence of a $k$-clique in a graph (see also Definition \ref{dfn:clique}) exhibits a sharp threshold with a window size of order $1/k^2$. In particular, if $k=\Theta((\log n)^x)$ for any $x>0$ then the window size of the $k$-clique function is of order $1/(\log n)^{O(1)}$ and therefore any desired extension of Theorem \ref{thm:conv} would imply that the $k$-clique $f$ is computable in $\mathcal{AC}_0$ on average for these values of $k$. Notably, computing $f$ for $\omega(1)=k=O(\log n)$ in $\mathcal{AC}_0$ for a worst-case graph is well-known to be impossible, and various average-case lower bounds are also known. Yet, to the best of our knowledge no such average-case lower bounds is applicable for the Erdos-Renyi measure (arguably, the simplest of all such measures on inputs) which is what we mean by ``on average'' in the present work.

\end{remark}

\section{Applications}\label{sec:apps}

In this section, we use Theorem \ref{thm:main_1} to obtain new ``average-case'' circuit lower bounds for two fundamental problems: the $k$-clique problem and 2-SAT. We need first a definition.
\begin{definition}\label{dfn:approx}
    Fix $f: \{0,1\}^N \rightarrow \{0,1\}$. We say that a circuit $\mathcal{C}: \{0,1\}^N \rightarrow \{0,1\}$ computes $f$ \emph{on average}, if for some constant $\xi \in (0,1/2)$ and $p_1=(p_1)_N, p_2=(p_2)_N \in (0,1)$ with $\E_{p_1}f = \xi$ and $\E_{p_2}f = 1-\xi,$ it holds that for any constant $\epsilon>0$ if $N$ is large enough, 
    $$  \max \{ \E_{p_1}|\mathcal{C}-f|,\E_{p_2}|\mathcal{C}-f|\} \leq \epsilon.$$
\end{definition}

\begin{remark}
Informally, we say that a circuit computes a Boolean function on average if it is equal to the Boolean function with high probability, for at least \emph{two values of} $p$ at the function's window. This way of defining the average-case guarantees for circuits based on their performance around criticality has already appeared before in the literature (see e.g., \cite{rossman2008constant,amano2010k,li2017ac} for subgraph isomorphism/containment problems). We note that the definition is also motivated by the random $k$-SAT literature where ``hard'' instances have been proposed to appear around the critical threshold \cite{selman1996generating, selman1996critical, mitchell1996some}.
\end{remark}

\subsection{Circuit lower bound for linear-sized cliques}\label{sec:apps_cliques}

Our first application is establishing new lower bounds for bounded depth circuits for the $k$-clique problem. 
\begin{definition}[The $k$-\textit{clique} problem]\label{dfn:clique}
    Let $n,k \in \N$ with $k \leq n.$ The $k$-clique Boolean function $f: \{0,1\}^{\binom{n}{2}} \rightarrow \{0,1\}$ equals to 1 if and only if the $n$-vertex undirected input graph $G \in \{0,1\}^{\binom{n}{2}}$ contains a $k$-clique. 
\end{definition}

We establish the following result.

\begin{theorem}\label{thm:clique}
    Let $n \in \mathbb{N}$ and $k \in \mathbb{N}$ with $k=\Theta(n)$. Any circuit $\mathcal{C}: \{0,1\}^{\binom{n}{2}} \rightarrow \{0,1\}$ with depth $d$ and size $s$ which computes $k$-clique on average  must satisfy either
        \begin{align}\label{eq:final_cl_0}
      d\geq (1/4-o(1))(\log \log n)^{-1}  \log n,
    \end{align}

 or
    \begin{align}\label{eq:final_cl1_0}
     s\geq \exp \left( (n^{(1/2-o(1))/(d+3)}\right).
    \end{align}

    In particular, $\mathcal{AC}_0$ cannot solve the $k$-clique problem for $k=\Theta(n)$ in $G(n,p)$ around the critical threshold.
\end{theorem}

\subsection{Circuit lower bound for 2-SAT}\label{sec:apps_sat}

We focus now on our lower bound for circuits computing 2-SAT.

  \begin{definition}[A 2-\textit{SAT} formula]
    Fix $n \in \N$. Let $n$ Boolean variables $x_1,\ldots,x_n$, which correspond to $2n$ literals $x_1,\ldots,x_n, \neg x_1,\ldots, \neg x_n.$ We call two literals $x,y$ strictly distinct if neither $x \not = y$, nor $\neg x \not = y.$ A 2-SAT formula is a conjunction of $m \in N$ distinct clauses, each of which is the disjunction of at most two strictly distinct literals.
 \end{definition}  

 \begin{definition}[The 2-\textit{SAT} Boolean function]
    Let $C_1,\ldots,C_N, i=1,\ldots,N=2n(n-1)$ be an ordering of the entire set of clauses (with strictly distinct literals) of a 2-SAT formula. Note that any 2-SAT formula can be encoded by $X \in \{0,1\}^N$ where $X_i=1$ if and only if clause $C_i$ is used in the formula. The 2-SAT Boolean function $f: \{0,1\}^N \rightarrow \{0,1\}$ equals 1 if and only if $X$ is satisfiable by some assignment on $x_i \in \{0,1\}, i=1,\ldots,n.$ 
 \end{definition}

We establish the following result.

\begin{theorem}\label{thm:sat}
    Let $n \in \mathbb{N}$. Any circuit $\mathcal{C}: \{0,1\}^{2n(n-1)} \rightarrow \{0,1\}$ with depth $d$ and size $s$ which computes 2-SAT on average must satisfy,
   either \begin{align}\label{eq:final_SAT_0}
      d \geq (1/6-o(1))(\log \log n)^{-1}  \log n.
    \end{align} or 
    \begin{align}\label{eq:final_SAT1_0}
     s \geq \exp \left( (n^{1/(3d+9)-o(1)}\right).
    \end{align}
    In particular, $\mathcal{AC}_0$ cannot solve random 2-SAT around the critical threshold.
\end{theorem}

\section{Circuit Lower bounds for statistical problems}

In this section, we use our main result to establish circuit lower bounds against statistical estimation problems. Specifically, we study the following natural family of Boolean statistical problems, where a statistician wishes to recover a hidden subset of $[N]$ from observing its union with a random subset of $[N].$

\begin{definition}[Hidden Subset Problem]\label{dfn:hidden}
Fix $N \in \N$ and an arbitrary known prior distribution $\mathcal{P}=\mathcal{P}_N$ over $\{0,1\}^{N}$. The \emph{Hidden Subset Problem}, parameterized by the prior $\mathcal{P},$ is the following task. For some noise level $p \in (0,1),$ and for a sample $S \sim \mathcal{P},$ the statistician observes $Y=S \lor \bX$ where $\bX \sim \P_p.$ The goal of the statistician is to construct an estimator $A=A(N,\mathcal{P},p): \{0,1\}^N \rightarrow \{0,1\}^{N}$ which achieves \emph{exact recovery}, i.e., for large enough $N$, $$ \P_{S \sim \mathcal{P}, \bX \sim \P_p} (A(Y=S \lor \bX)=S)  \geq 0.9.$$
\end{definition}We note that, of course, the choice of $0.9$ is arbitrary and can be replaced by any constant less to one. As we explain in Section \ref{sec:apps_pc}, the planted clique model can be modelled as a Hidden Subset Problem and as explained in \cite{mossel2023sharp} much more generality is possible.

Notice that the larger $p$ is, the harder is to recover $S$ in the Hidden Subset Problem. Hence, we interpret $1/p$ as the noise level in this problem. We now proceed with the so-called, information-theoretic threshold of the problem.

\begin{definition}[Information-theoretic threshold]\label{dfn:it_threshold}
    For the Hidden Subset Problem $\mathcal{P}$, we call $p_{\mathrm{IT}}=p_{\mathrm{IT}}(\mathcal{P})$ the supremum over all $p$ for which some estimator $A: \{0,1\}^N \rightarrow \{0,1\}^{N}$ achieves \emph{exact} recovery of $S$.
\end{definition}Note that in statistical terms, this can be phrased as matching the performance of the optimal maximum a posteriori estimator. Also observe that by definition $p_{\mathrm{IT}}$ may be depending on $N.$

As we explained in the introduction, sometimes estimation settings exhibit a sharp threshold as the noise level changes, known as the All-or-Nothing (AoN) phenomenon. As we mentioned, in this context the role of the noise level is played by $1/p$ and a recent work \cite{mossel2023sharp} has extensively studied under what conditions AoN for the Hidden Subset Problem appears as $p$ increases.

We proceed with defining the following slight variant of the AoN phenomenon adapted to the exact recovery task. The variation is that AoN has been traditionally defined for the mean-squared-error metric, but in our context, the exact recovery criterion (i.e., 0-1 loss) is more technically appropriate. 

Recall that by Definition \ref{dfn:it_threshold}, for any $\epsilon=\epsilon_N>0$ if $p <(1-\epsilon)p_{\mathrm{IT}},$ then \emph{some} algorithm $A: \{0,1\}^N \rightarrow \{0,1\}^{N}$ achieves for large enough $N$, \begin{align}\label{eq:all} \P_{S \sim \mathcal{P}, \bX \sim \P_p} (A(Y=S \lor \bX)=S)  \geq 0.9.\end{align} We informally call \eqref{eq:all} as $A$ achieving ``all''. The AoN phenomenon is the following.
\begin{definition}[The All-or-Nothing (AoN) Phenomenon for exact recovery]
We say that the Hidden Subset Problem $\mathcal{P}=\mathcal{P}_N$ exhibits the All-or-Nothing (AoN) Phenomenon with window size $\epsilon=\epsilon(\mathcal{P}_N) \in (0,1)$ if the following holds.

If $p>(1+\epsilon)p_{\mathrm{IT}},$ then \emph{all} algorithms $A: \{0,1\}^N \rightarrow \{0,1\}^{N}$ satisfy for large enough $N$, \begin{align}\label{eq:nothing} \P_{S \sim \mathcal{P}, \bX \sim \P_p} (A(Y=S \lor \bX)=S)  \leq 0.1.\end{align}
\end{definition}Notice that ``nothing'' (informally referring to \eqref{eq:nothing}) here consists of a sudden drop in the probability of exact recovery by any algorithm. The choice of $0.1$ is arbitrary and can be replaced by any constant bigger than zero. 

Before proceeding with the connection between the AoN phenomenon and circuit complexity, we clarify what we mean for an algorithm to \emph{solve} (to information-theoretic optimality) the Hidden Subset Sum problem of interest.
\begin{definition}(A successful algorithm)
    We say that an algorithm $A: \{0,1\}^N \rightarrow \{0,1\}^{N}$ \emph{solves} the Hidden Subset Problem $\mathcal{P},$ if and only if it achieves \emph{exact} recovery of $S$ whenever $p < p_{\mathrm{IT}}.$ 
\end{definition}

Using Theorem \ref{thm:main_1}, we obtain the following size lower bounds for bounded depth circuits solving an estimation task exhibiting the AoN phenomenon.

\begin{theorem}\label{thm:circuit_aon}
    Let $N \in \mathbb{N},$ and prior $\mathcal{P}=\mathcal{P}_N$ on $\{0,1\}^N.$ We assume the Hidden Subset Problem parametrized by $\mathcal{P}$ exhibits the AoN phenomenon for some window size $\epsilon=\epsilon(\mathcal{P})$ where we assume $\epsilon=o((1-p_{\mathrm{IT}})/\log 1/\beta)$ and $\beta=N^{-\Omega(1)}$ for $\beta=\min\{p_{IT},1-p_{IT}\}$. Let $\mathcal{C}: \{0,1\}^{N} \rightarrow \{0,1\}^{N}$ be a circuit of depth $d$ and size $s$ which solves the Hidden Subset problem. Then for some universal constants $c_1,c_2,c_3>0,$ it must hold either, 
 \begin{align}\label{eq:final_000_aon}
      d \geq \frac{\log \frac{1-p_{\mathrm{IT}}}{\epsilon\log 1/\beta}}{2\log \log N }-6.
    \end{align} or
    \begin{align}\label{eq:final_01_aon}
     s \geq c_2\exp \left( c_3\left(\frac{1-p_{\mathrm{IT}}}{\epsilon \log 1/\beta}\right)^{1/(d+6)}\right)
    \end{align}
\end{theorem}

The proof of the theorem is in Section \ref{sec:omitted}.

\subsection{Application: the Planted $k$-Clique Problem}\label{sec:apps_pc}

To show the applicability of our result Theorem \ref{thm:circuit_aon}, we now show how it implies circuit lower bounds for a variant of the well-studied estimation task called Planted Clique. We first define it as a special case of a Hidden Subset Problem.

  \begin{definition}[The $k$-planted clique problem]
    Fix $n$ vertices and some $k=k_n \in \N$. Let $N=\binom{n}{2}$ be indexing all the possible undirected edges between the vertices. The $k$-clique problem is the Hidden Subset Problem where $S \sim \mathcal{P}$ is the prior on $\{0,1\}^N$ which first chooses $k$ vertices out of $n$ at random and then sets $S \in \{0,1\}^{\binom{n}{2}}$ to be the indicator of the edges defined by the $k$ vertices. Note that $S$ corresponds now to the edges of $k$-clique and $Y=S\lor \bX \in \{0,1\}^{\binom{n}{2}}, \bX \sim \P_p,$ correspond to a graph where $S$ is a $k$-planted clique and each other edge appears independently with probability $p.$
 \end{definition}

With this interpretation, we establish the following result on the ability of circuits to estimate the planted clique. As we explained in the Introduction we focus on the case where $k$ is super logarithmic in $n$ in which case the information-theoretic threshold $p_{\mathrm{IT}}$ is converting to 1 with $n.$ This assumption is crucial as it allows to establish an AoN phenomenon with a sufficiently small window. Moreover, we remind the reader that by solving the planted clique we mean recovering its vertices with probability at least $0.9$ for all $p<p_{\mathrm{IT}}.$

\begin{theorem}\label{thm:planted-clique}
    Let $n,k=k_n \in \mathbb{N}$ with $ (\log n)^{\omega(1)}=k \leq n^{1/3-\Omega(1)}$. Any circuit $\mathcal{C}: \{0,1\}^{\binom{n}{2}} \rightarrow \{0,1\}^{\binom{n}{2}}$ of depth $d$ and size $s$ which solves the $k$-clique planted problem must satisfy for some constants $c_1,c_2,c_3>0$
 \begin{align}\label{eq:final_0_pc}
      d \geq (1/2-o(1))(\log \log n)^{-1}  \log k.
    \end{align} or
    \begin{align}\label{eq:final_1_pc}
     s \geq c_1\exp \left( c_2k^{1/(d+6)}\right)
    \end{align}
    In particular, for $ (\log n)^{\omega(1)}=k \leq n^{1/3-\Omega(1)}$, $\mathcal{AC}_0$ cannot solve the $k$-planted clique problem.
\end{theorem}

The proof naturally combines Theorem \ref{thm:circuit_aon} and that planted $k$-clique exhibits the AoN phenomenon for exact recovery, when $ (\log n)^{\omega(1)}=k \leq n^{1/3-\Omega(1)}$ (see Lemma \ref{lem:sharp_transition_pc}).

\section{Proof of our main result, Theorem \ref{thm:main_1}}

We will approach the cases $p_c \leq 1/2$ and $p_c>1/2$ in a slightly different manner. For this reason, we establish the following lemma from which Theorem \ref{thm:main_1} follows.

\begin{lemma}\label{lem:interm}
We make the assumptions of Theorem \ref{thm:main_1} with the additional assumption that $p_c \leq 1/2$.  

Then for some constants $c_1,c_2>0$
        either,
 \begin{align}\label{eq:final_00_lem}
      d \geq \frac{1}{2\log \log N}\log \frac{\delta}{\epsilon \log 1/p_c} -3.
    \end{align} or
    \begin{align}\label{eq:final_01_lem}
     s \geq c_1\exp \left( c_2\left(\frac{\delta}{\epsilon \log 1/p_c}\right)^{1/(d+3)}\right).
    \end{align}
    
\end{lemma}

\subsection{Proof that Lemma \ref{lem:interm} implies Theorem \ref{thm:main_1}} Here we show that Theorem \ref{thm:main_1} follows from Lemma \ref{lem:interm} as follows.

If $p_c \leq 1/2$ then it follows as the lower bound of \eqref{eq:final_00_lem} is stronger than \eqref{eq:final_0000} and \eqref{eq:final_01_lem} is stronger than \eqref{eq:final_001}. 

If $p_c \geq 1/2,$ we argue as follows. Notice that the circuit $\tilde{\mathcal{C}}(X):=\mathcal{C}(\neg X),$ that applies $\mathcal{C}$ on the negations of the input $X$, has exactly the same depth and size with $\mathcal{C}$. Moreover, by the obvious coupling for any $p>0,$ $$\E_{1-p} \tilde{\mathcal{C}}=\E_{p}\mathcal{C}.$$Hence if
\begin{align}\label{eq:diff_origin}
        |\E_{(1+\epsilon)p_c} \mathcal{C}-\E_{(1-\epsilon)p_c} \mathcal{C}| \geq \delta,
    \end{align} for $\epsilon':=\epsilon p_c/(1-p_c)$ we have,
\begin{align}\label{eq:diff_neg}
        |\E_{(1+\epsilon')(1-p_c)} \tilde{\mathcal{C}} -\E_{(1-\epsilon')(1-p_c)}\tilde{\mathcal{C}} | \geq \delta
    \end{align}since $1-(1+\epsilon)p_c=(1-\epsilon')p_c$ and $1-(1-\epsilon)p_c=(1+\epsilon')(1-p_c)$. Therefore, $\tilde{\mathcal{C}}$ has a sharp threshold at the threshold $1-p_c \leq 1/2$, with window size $\epsilon'$ and jump size $\delta.$ Note that by our assumption on $\epsilon=o((1-p_c)/\log(1/(1-p_c))$ and $p_c \geq 1/2$, $$\epsilon'=\epsilon p_c/(1-p_c)=o((1-(1-p_c))/\log 1/(\min\{p_c,1-p_c\})).$$Therefore by Lemma \ref{lem:interm} we conclude for $p_c \geq 1/2$ that it must hold (as $C, \tilde{C}$ have the same size and depth),
    \begin{align}\label{eq:final_00_lem_1}
      d=d \geq \frac{\log \frac{\delta(1-p_c)}{\epsilon p_c \log 1/(1-p_c)}}{2\log \log N} -3.
    \end{align} or
    \begin{align}\label{eq:final_01_lem_2}
     s \geq c_1\exp \left( c_2\left(\frac{\delta(1-p_c)}{\epsilon p_c \log 1/(1-p_c)}\right)^{1/(d+3)}\right).
    \end{align}Now Theorem \ref{thm:main_1} in the case $p_c \geq 1/2$ also follows as the lower bound of \eqref{eq:final_00_lem_1} is stronger than \eqref{eq:final_0000} and \eqref{eq:final_01_lem_2} is stronger than \eqref{eq:final_001}. 
    
\subsection{Proof of Lemma \ref{lem:interm}} We now focus on proving Lemma \ref{lem:interm}, and in particular operate under the assumption \begin{align} \label{eq:pc} p_c \leq 1/2.\end{align}

In order to employ tools from Boolean Fourier analysis it will be useful to work exclusively in the regime of constant $p$ where $p_c=\Theta(1)$. However,  this may not be the case in many interesting applications (see e.g., random 2-SAT \cite{bollobas2001scaling}). 
For this reason, we first employ the following lemma which allows us to always ``move'' the phase transition point $p_c$ to the $\Theta(1)$ regime.  Indeed this can be achieved
by adding only one extra layer to the circuit, of size \emph{linear in $N$ and $\log 1/p_c$}. The lemma below is an adopted and simplified version of a similar construction in~\cite{gamarnik2020low}, together with a necessary stability analysis of the result to small perturbations.

%To be more specific, the following change of measure via an $1$-depth circuit holds.
\begin{lemma}[``Debiasing'' layer]\label{lem:circuit}
   Let $N \in \mathbb{N}$ and arbitrary $p_0=(p_0)_N \leq 1/2.$ There exists a depth-1 and size $O(N\log 1/p_0)$ circuit $\Phi: \{0,1\}^{N \ceil{\log 1/p_0}}  \rightarrow \{0,1\}^{N}$ and some $p_1=(p_1)_N\in [1/2,1/\sqrt{2})$ such that 
   \begin{itemize}
       \item \begin{align}\label{eq:dist_1}
        \Phi\left(\mathrm{Bern}(p_1)^{\otimes N\ceil{\log 1/p_0}}\right)\overset{d}{=} \mathrm{Bern}(p_0)^{\otimes N},
    \end{align}and,

    \item for any $0<\gamma=\gamma_N=o(1/\log 1/p_0),$ there exists $0<r_N=\Theta(\gamma \log 1/p_0)=o(1)$ with
    \begin{align}\label{eq:dist_2}
        \Phi\left(\mathrm{Bern}(p_1-r_n)^{\otimes N \ceil{\log 1/p_0}}\right)\overset{d}{=}  \mathrm{Bern}(p_0(1-\gamma))^{\otimes N}.
    \end{align}
   \end{itemize}

\end{lemma}
The proof of this lemma is found in Section~\ref{sec:omitted}.
Now, given \eqref{eq:pc} and our assumptions on Theorem \ref{thm:main_1}, we choose $p_0:=(1+\epsilon)p_c$ and apply Lemma \ref{lem:circuit} to obtain the following. For some $p_1=(p_1)_N \in [1/2,1/\sqrt{2}),$ and some 1-depth circuit $\Phi_d: \{0,1\}^{N\ceil{\log 1/(1+\epsilon)p_c}} \rightarrow \{0,1\}^N$ of size $O(N \log 1/p_c)$ it holds via \eqref{eq:dist_1},

 \begin{align}\label{eq:dist_11}
        \E_{p_1}\mathcal{C}\circ \Phi =\E_{(1+\epsilon)p_c} \mathcal{C}. 
    \end{align}Also, applying \eqref{eq:dist_2} for $\gamma=1-\frac{1-\epsilon}{1+\epsilon}=\Theta(\epsilon)$ we obtain for some $0<r_N=\Theta(\epsilon \log 1/p_c)=o(1)$
    \begin{align}\label{eq:dist_22}
   \E_{p_1-r_N}\mathcal{C} \circ \Phi =\E_{(1-\epsilon)p_c} \mathcal{C}.
    \end{align}
In particular, given the sharp threshold for $\mathcal{C}$, the following sharp threshold holds for $\mathcal{C} \circ \Phi$,
     \begin{align}\label{eq:diff}
        |\E_{p_1}\mathcal{C} \circ \Phi -\E_{p_1-r_N} \mathcal{C} \circ \Phi| \geq \delta.
    \end{align}

    Since $r_N=o(1),$ for large enough $N$ it holds $(p_1-r_N,p_1) \subseteq [1/3,2/3].$
Hence, by mean value theorem, using \eqref{eq:diff} for large enough $N$ it must also hold,
\begin{align}\label{eq:deriv_old}
      \max_{p \in [1/3,2/3]}  |\frac{d}{dp}\E_{p} \mathcal{C} \circ \Phi |\geq \delta/r_N. 
    \end{align}

    Since $r_N=\Theta(\epsilon \log 1/p_c)$ we conclude that for some universal constant $c_1>0,$
    \begin{align}\label{eq:deriv_0}
      \max_{p \in [1/3,2/3]}  |\frac{d}{dp}\E_{p} \mathcal{C} \circ \Phi  | \geq c_1\delta/(\epsilon \log 1/p_c). 
    \end{align}
Note that $C \circ \Phi$ now has a sharp threshold with a ``constant'' threshold. Next, we use the following two standard lemmas (their proofs are in Section \ref{sec:omitted}). Recall the definition of total influence from Section \ref{sec:notation}.

\begin{lemma}[One-sided ``Russo-Margulis Lemma'']\label{lem:russo}
    For any Boolean function $f: \{0,1\}^N \rightarrow \{0,1\},$
    $$|\frac{d}{dp} \E_p f| \leq (p(1-p))^{-1}I_p(f).$$
\end{lemma}

\begin{lemma}[Extension of the LMN Theorem to arbitrary $p$.]\label{lem:LMN}
   There exists a constant $c_0>0$ such that, if $\mathcal{C}: \{0,1\}^N \rightarrow \{0,1\}$ is a Boolean circuit of depth $D$ and size $S$ then $$I_p(\mathcal{C}) \leq c_0\left(\frac{10 \log (S )}{p(1-p)}\right)^{D+2}.$$
\end{lemma}

Note that the circuit $\mathcal{C} \circ \Phi$ has size $s+O(N\log 1/p_c)$ and depth $d+1$. Hence applying Lemmas \ref{lem:russo} and \ref{lem:LMN} for $f=\mathcal{C} \circ \Phi$ we conclude that for some universal constant $c_2>0,$
\begin{align}\label{eq:deriv_2}
       \max_{p \in [1/3,2/3]}|\frac{d}{dp}\E_{p} \mathcal{C} \circ \Phi |\leq \left(c_2 \log \left(s+c_2N\log 1/p_c\right))\right)^{d+3}. 
    \end{align}
Combining \eqref{eq:deriv_2} with \eqref{eq:deriv_0} we conclude for some universal constant $c_3,c_4>0,$

\begin{align}\label{eq:final_1}
      2c_4\exp \left( c_3(\frac{\delta}{\epsilon \log 1/p_c})^{1/(d+3)}\right) \leq  s+N\log 1/p_c.
    \end{align}

    Hence, it must hold
    \begin{align}\label{eq:final_2}
      \max\{s,N\log 1/p_c\} \geq c_4\exp \left( c_3(\frac{\delta}{\epsilon \log 1/p_c})^{1/(d+3)}\right).
    \end{align}

    Hence, either
\begin{align}\label{eq:final_case_1}
      N \log 1/p_c \geq c_4\exp \left( c_3(\frac{\delta}{\epsilon \log 1/p_c})^{1/(d+3)}\right).
    \end{align}which implies

    \begin{align*}
    \log(N \log 1/p_c/c_4) \geq c_3(\frac{\delta}{\epsilon \log 1/p_c})^{1/(d+3)}
    \end{align*}and therefore

     \begin{align*}
      \log (\log(N \log 1/p_c/c_4)/c_3) \geq  (d+3)^{-1} \log(\frac{\delta}{\epsilon \log 1/p_c}),
    \end{align*}from which we conclude for a sufficiently large $c_5>0,$
    \begin{align*}
      d\geq \frac{\log \frac{\delta}{\epsilon \log 1/p_c}}{\log \log (N\log 1/p_c) +c_5} -3.
    \end{align*} Observe that for large enough $N$, $\log \log (N\log 1/p_c) +c_5 \leq 3/2\log \log (N\log 1/p_c) \leq 2 \log \log N,$ where in the last step we use $p_c =N^{-\Omega(1)}.$ The inequality \eqref{eq:final_00_lem} thus follows.

If \eqref{eq:final_case_1} does not hold then \eqref{eq:final_2} implies,
    \begin{align}\label{eq:final_case_2}
     s\geq c_4\exp \left( c_3(\frac{\delta}{\epsilon \log 1/p_c})^{1/(d+3)}\right).
    \end{align}
    The inequality \eqref{eq:final_01_lem} then follows.

\section{Proof of the circuit lower bound for clique, Theorem \ref{thm:clique}}

The proof combines Theorem \ref{thm:main_1} with the following sharp threshold. We call $\mathcal{P}_k$ the set of $n$-vertex undirected graphs containing a $k$-clique. 
\begin{lemma}[Sharp transition -- linear cliques in random graphs]\label{lem:sharp_transition} Fix any $k=\Theta(n)$ and any constant $\xi \in (0,1/2).$ For $p=1-d/n$, let $ d=d_{n,k}>0$ such that $\PP_p \left(G \in \mathcal{P}_{k}\right)=1-\xi.$  Then,
    \begin{itemize}
        \item[(a)] $d=\Theta(1).$
        \item[(b)] Fix any $0<\gamma<1/2.$\\
        For $p=1-d/n-1/n^{3/2-\gamma}=(1-\Theta(1/n^{3/2-\gamma}))(1-d/n),$ and $n$ sufficiently large, $\PP_p \left(G \in \mathcal{P}_{k}\right)\leq \xi$.
    \end{itemize} 
\end{lemma}The proof of the Lemma is in Section \ref{sec:omitted}.

Given Lemma \ref{lem:sharp_transition} the proof proceeds in a straightforward manner as follows.
Assume that for some $k=\Theta(n),$ a circuit $\mathcal{C}$ computes the $k-$clique Boolean function $f$ on average. Then for some constant $\xi \in (0,1/2),$ and $p_1,p_2 \in (0,1)$ such that $\E_{p_1}f=\xi, \E_{p_2}f=1-\xi,$ it holds also \begin{align*}\E_{p_2}\mathcal{C}-\E_{p_1}\mathcal{C} \geq \E_{p_2}f-\E_{p_1}f -o(1) \geq 1-2\xi-o(1).
\end{align*}By Lemma \ref{lem:sharp_transition} for this value of $\xi$, we have $p_2=1-\Theta(1/n)$ and $|p_1-p_2| \leq n^{-3/2+\gamma}$ for any small constant $\gamma>0.$ Hence $\mathcal{C}$ exhibits a sharp threshold at the location $p_c=1-\Theta(1/n),$ window size $\epsilon_N=o(n^{-3/2+\gamma})$ and jump size $\delta=\Theta(1).$ 

Notice $$\frac{\delta(1-p_c)}{\epsilon \log 1/\beta}=\Omega(n^{1/2-\gamma}/\log n).$$ 
Hence, using Theorem \ref{thm:main_1} for $N=\binom{n}{2}$, we conclude that for any arbitrarily small constant $\gamma>0$ either 

\begin{align}\label{eq:final_cl}
      d(\mathcal{C}) \geq (1/4-\gamma)(\log \log n)^{-1}  \log n.
    \end{align} or 
    \begin{align}\label{eq:final_cl1}
     s(\mathcal{C}) \geq \exp \left( (n^{(1/2-2\gamma)/(d(\mathcal{C})+3)}\right).
    \end{align}This concludes the proof.

\section{Proof of the circuit lower bound for random 2-SAT, Theorem \ref{thm:sat}}
The proof combines Theorem \ref{thm:main_1} with following sharp threshold from \cite{bollobas2001scaling}. Let $\mathcal{P}_{\mathrm{2-SAT}}$ be the set of all 2-SAT formulas which are satisfiable and $f$ the Boolean function representing it.

Note that in our notation, $\bX \sim \P_p$ corresponds to a random 2-SAT formula where every clause is included into the formula 
independently with probability $p$. The following sharp threshold holds.
\begin{lemma}[Sharp transition -- scaling window of 2-SAT \cite{bollobas2001scaling}]\label{lem:sharp_transition_SAT} 
For any constant $\xi \in (0,1/2)$ there exists a large enough constant $c>0$, 
    \begin{itemize}
        \item[(a)] For $p=1/(2n)-c/n^{4/3}$, $$\PP_{p} \left(\bX \in \mathcal{P}_{\mathrm{2-SAT}}\right)\geq 1-\xi.$$

        \item[(b)] For $p=1/(2n)+c/n^{4/3}$, $$\PP_{p} \left(\bX \in \mathcal{P}_{\mathrm{2-SAT}}\right)\leq \xi.$$
    \end{itemize} 
\end{lemma}

Given Lemma \ref{lem:sharp_transition_SAT} the proof proceeds as follows.
Assume that a circuit $\mathcal{C}$ computes the Boolean random 2-SAT $f$ on average. Then for some constant $\xi \in (0,1/2),$ and $p_1,p_2 \in (0,1)$ such that $\E_{p_1}f=\xi, \E_{p_2}f=1-\xi,$ it holds also \begin{align*}\E_{p_2}\mathcal{C}-\E_{p_1}\mathcal{C} \geq \E_{p_2}f-\E_{p_1}f -o(1) \geq 1-2\xi-o(1).
\end{align*}
Then, Lemma \ref{lem:sharp_transition_SAT}, implies $|p_1-p_2| \leq n^{-4/3}$ and $p_1,p_2=\Theta(1/n)$. Hence $\mathcal{C}$ must exhibit a sharp threshold at the location $p_c=\Theta(1/n)$, window size $\epsilon_n=\Theta(n^{-1/3})$ and jump size $\delta=\Theta(1).$ 

Notice $$\frac{\delta(1-p_c)}{\epsilon \log 1/\beta}=\Theta(n^{1/3}/\log n)=n^{1/3-o(1)}.$$ Hence, using Theorem \ref{thm:main_1} for $N=2n(n-1)$, we conclude that either for all small constant $\eta>0$ and large enough $n$,

\begin{align}\label{eq:final_SAT}
      d \geq (1/6-\eta)(\log \log n)^{-1}  \log n.
    \end{align} or for all small constant $\eta>0$ and large enough $n$,
    \begin{align}\label{eq:final_SAT1}
     s \geq \exp \left( (n^{1/(3d+9)-\eta}\right).
    \end{align}This concludes the proof.

 \section{Proof of the circuit lower bound for Planted Clique, Theorem \ref{thm:planted-clique}}

The proof combines Theorem \ref{thm:circuit_aon} with the following new sharp threshold phenomenon on the exact recovery of the planted clique.
\begin{lemma}[Sharp transition -- planted clique]\label{lem:sharp_transition_pc} Suppose for some constant $\delta>0,$ $(\log n)^{\omega(1)}=k \leq n^{1/3-\delta}.$ Then the $k$-planted clique problem exhibits the All-or-nothing phenomenon for exact recovery at some $p_{\mathrm{IT}} =1-\Theta(\log n/k) $ with window size $\epsilon=\Theta(1/k^2)$. 
\end{lemma}The proof of the Lemma is in Section \ref{sec:omitted}.

Given Lemma \ref{lem:sharp_transition_pc} the proof proceeds as follows.
Let $(\log n)^{\omega(1)}=k \leq n^{1/3-\Omega(1)},$ and consider the $k$-planted clique problem. Then, given Lemma \ref{lem:sharp_transition_pc}, it must exhibit the All-or-Nothing phenomenon for exact recovery with window size $\epsilon_N=\Theta(1/k^2)$ at some $p_{\mathrm{IT}} =1-\Theta(\log n/k).$ Moreover, $\beta=\min\{1-p_{\mathrm{IT}},p_{\mathrm{IT}}\}=\Theta(\log n/k).$ It holds $(1-p_{\mathrm{IT}})/\log 1/\beta=\Omega(1/k)$. Hence $$\epsilon=o((1-p_{\mathrm{IT}})/\log 1/\beta)$$and in particular, $$\frac{1-p_{\mathrm{IT}}}{\epsilon \log 1/\beta}=\Omega(k).$$ Hence, using Theorem \ref{thm:main_1} for $N=\binom{n}{2}$, and that $(\log n)^{\omega(1)}=k \leq n^{1/3-\Omega(1)},$  we conclude that either

\begin{align}\label{eq:final_00_pc}
      d(\mathcal{C}) \geq (1/2-o(1))(\log \log n)^{-1}  \log k.
    \end{align} or
    \begin{align}\label{eq:final_01_pc}
     s(\mathcal{C}) \geq c_1\exp \left( c_2k^{1/(d(\mathcal(C)+6)}\right)
    \end{align}  This concludes the proof.

\section{Omitted Proofs}\label{sec:omitted}

  \begin{proof}[Proof of Lemma \ref{lem:circuit}]

For input $X \in \{0,1\}^{N\ceil{\log 1/p_0}}$, we split the coordinates of $X$ into the $N$ disjoint blocks of $\ceil{\log 1/p_0}$ consecutive bits which we denote by $X_1,\ldots,X_{N} \in \{0,1\}^{\ceil{\log 1/p_0}}$, i.e., $X=(X_1,\ldots,X_N)$. Then, for $i=1,2,\ldots, N$, we set $$\Phi(X)_i=\land_{j=1}^{\ceil{\log 1/p_0}}  (X_i)_j.$$ Clearly $\Phi$ can be modeled by a depth 1 and $O(N\log 1/p_0)$-size Boolean circuit is obvious.

Now, notice that $\Phi(X)_i=1$ if and only if $(X_i)_j=1$ for all $ 1 \leq  j \leq \ceil{\log 1/p_0}.$ Hence, for any $p \in (0,1)$ if $X_i \sim \mathrm{Bern}(p)^{\otimes \ceil{\log 1/p_0}}$ then $\Phi(X)_i=1$ with probability $p^{\ceil{\log 1/p_0}}$, i.e., it holds
\begin{align}\label{eq:dist}
        \Phi\left(\mathrm{Bern}(p)^{\otimes N\ceil{\log 1/p_0}}\right)\overset{d}{=} \mathrm{Bern}(p^{\ceil{\log 1/p_0}})^{\otimes N}.
    \end{align} To prove now \eqref{eq:dist_1}, \eqref{eq:dist_2} we choose $p_1=p_0^{1/\ceil{\log 1/p_0}}$. Since $p_0 \leq 1/2,$ it holds $$p_1=2^{- \log 1/p_0/\ceil{\log 1/p_0}}\in [1/2,1/\sqrt{2})=\Theta(1).$$ Also, of course, $p_1^{\ceil{\log 1/p_0}}=p_0$ and therefore \eqref{eq:dist_1} holds. 
    
    For \eqref{eq:dist_2} we set $p=p_1-r_n,$ to obtain from \eqref{eq:dist} that it suffices the solution $r_n>0$ of
\begin{align}\label{eq:inv}
 (p_1-r_n)^{\ceil{\log 1/p_0}}=p_0(1-\gamma)   
\end{align} to satisfy $r_n=\Theta(\gamma \log 1/p_0).$ But \eqref{eq:inv} is equivalent with $$p_1-r_n=p_1(1-\gamma)^{1/\ceil{\log 1/p_0}}.$$Since $\gamma \log 1/p_0=o(1)$ and $p_1=\Theta(1)$, the last equation implies $$r_n=p_1-p_1(1-\Theta(\gamma \log 1/p_0))=\Theta(\gamma \log 1/p_0).$$This completes the proof.
 
    \end{proof}

\begin{proof}[Proof of Lemma \ref{lem:russo}]
    Viewing the $p$-biased measure on $\{0,1\}^N$ as the diagonal case of the $(p_1,p_2,\ldots,p_N)$-biased measure on $\{0,1\}^N$ we have by chain rule,
    $$\frac{d}{dp} \E_p f=\sum_{i=1}^N \frac{d}{dp_i} \E_{(p_1,\ldots,p_N)} f |_{(p_1,\ldots,p_N)=(p,\ldots,p)}.$$Now, by direct calculations, \begin{align*}|\frac{d}{dp_i} \E_{(p_1,\ldots,p_N)} f |_{(p_1,\ldots,p_N)=(p,\ldots,p)}|& =|\E_{p_{-i}}[f(x_{-i},1)-f(x_{-i},0)]|_{(p_1,\ldots,p_N)=(p,\ldots,p)} | \\
    & \leq \E_{p_{-i}}[|f(x_{-i},1)-f(x_{-i},0)|]|_{(p_1,\ldots,p_N)=(p,\ldots,p)}.\end{align*}

    Now notice that for any $x,$ $f(x_{-i},1)-f(x_{-i},0) \in \{-1,0,1\}.$ Hence, it must hold 
    $$|f(x_{-i},1)-f(x_{-i},0)| = (f(x_{-i},1)-f(x_{-i},0))^2.$$Therefore,

    \begin{align*}|\frac{d}{dp_i} \E_{(p_1,\ldots,p_N)} f |_{(p_1,\ldots,p_N)=(p,\ldots,p)}| & \leq \E_{p_{-i}}[(f(x_{-i},1)-f(x_{-i},0))^2]|_{(p_1,\ldots,p_N)=(p,\ldots,p)}    \\ 
    &= (p(1-p))^{-1} (I_i)_p(f).\end{align*}
    The result follows.
\end{proof}

We also need the following variant of the LMN theorem.

\begin{proof}[Proof of Lemma \ref{lem:LMN}]
    Given the standard orthonormal basis under the $p$-biased measure $\{\chi_{S}, S \subseteq [N]\}$, we have for any $f: \{0,1\}^N \rightarrow \{0,1\}$, if $\{\hat{f}(S), S \subseteq [N]\}$ is it's Fourier coefficients, it must hold $$I_p(f)=\sum_{S \subseteq [N]}|S|\hat{f}(S)^2.$$

    Now, by an extension of the LMN theorem \cite[Lemma 9]{furst1991improved} (see also \cite[Lemma 3.2]{bogdanov2015homomorphic}) to $p$-biased measures, we have that for the Boolean circuit $\mathcal{C},$ it holds that for any integer $k \geq 1$,
    $$ \sum_{S \subseteq [N], |S|>k} \hat{C}(S)^2 \leq S 2^{-p(1-p)k^{1/(D+2)}/5}.$$
Hence for some $k_0=O\left(\left(5 \frac{\log (NS)}{p(1-p)}\right)^{D+2}\right),$
$$ \sum_{S \subseteq [N], |S|>k_0} \hat{C}(S)^2 \leq 1/N.$$
Hence,

\begin{align*}
  \sum_{S \subseteq [N]} |S|\hat{C}(S)^2 &=   \sum_{S \subseteq [N], |S| \leq k_0} |S|\hat{C}(S)^2+\sum_{S \subseteq [N], |S|>k_0} |S|\hat{C}(S)^2 \\
  & \leq k_0\sum_{S \subseteq [N]} \hat{C}(S)^2+N \sum_{S \subseteq [N], |S|>k_0} \hat{C}(S)^2\\
  & \leq k_0+1 =O\left(\left(5 \frac{\log (NS)}{p(1-p)}\right)^{D+2}\right).
\end{align*}
The proof is complete.
\end{proof}

\begin{proof}[Proof of Lemma \ref{lem:sharp_transition}]
We equivalently focus on the graph property $\mathcal{I}_{k}$ of a graph including an independent set of size $k$; we can do this  since for any $p$, $\PP_p \left(G \in \mathcal{P}_{k}\right)=\PP_{1-p} \left(G \in \mathcal{I}_{k}\right).$ 

We define $\alpha(G)$ to be the independence number of an $n$-vertex graph $G$. Recall that $d>0$ is now a solution to 
 \begin{align}\label{eq:sol}
     \PP_{d/n} \left(G \in \mathcal{I}_{k}\right)=\PP_{d/n} \left(\alpha(G) \geq k\right)=1-\xi.
 \end{align} 
 
 The fact that for $k=\Theta(n),$ $d=\Theta(1)$ follows by standard results \cite{frieze1990independence}. 
 
 A folklore application of Azuma-Hoeffding concentration inequality for the vertex exposure martingale (see e.g., \cite[Section 7]{alon2016probabilistic}) implies that for any $\epsilon>0,$
 $$\PP_{d/n} \left(|\alpha(G)/n-\E_{G \sim G(n,d/n)}\alpha(G)/n| \geq n^{-1/2+\epsilon/2} \right)\leq 2\exp\left(-n^{\epsilon}/2\right). $$Combining with \eqref{eq:sol} we conclude $|k/n-\E_{G \sim G(n,d/n)}\alpha(G)/n| \leq n^{-1/2+\epsilon/2}$ and therefore
  \begin{align}
  \label{eq:conc}
    \PP_{d/n} \left(\alpha(G)/n \leq k/n+2n^{-1/2+\epsilon/2} \right)\geq 1-2\exp\left(-n^{\epsilon}/2\right).  
  \end{align}

    Notice that by taking the union of an instance of a $G(n,p)$ and an independent instance of a $G(n,(q-p)/(1-p))$ we arrive at an instance of a $G(n,q).$ Hence, for $G_1 \sim G(n,d/n)$ and independent $G_2 \sim G(n,n^{-3/2+\epsilon}/(1-d/n))$ we have $G=G_1\cup G_2$ is a sample of $G(n,d/n+n^{-3/2+\epsilon})$.

    Hence, combining with \eqref{eq:conc},

\begin{align}
    \PP_{d/n+n^{-3/2+\epsilon}} \left(G \in \mathcal{I}_k\right)&=\PP_{d/n+n^{-3/2+\epsilon}} \left(\alpha(G) \geq k\right)\nonumber\\
    &\leq \PP_{d/n+n^{-3/2+\epsilon}} \left(\alpha(G) \geq k, \alpha(G_1)/n \leq k/n+2n^{-1/2+\epsilon/2} \right)+2\exp\left(-n^{\epsilon}/2\right)\nonumber\\
    & \leq \PP_{d/n+n^{-3/2+\epsilon}} \left( \alpha(G_1)-\alpha(G) \leq 2n^{1/2+\epsilon/2} \right)+2\exp\left(-n^{\epsilon}/2\right) \label{eq:ineq}
\end{align}

 Now, observe that any maximal independent set of an instance of $G_1$ almost surely contains all the isolated nodes of $G_1$, a set we denote by $\mathcal{R}(G_1)$. Recall that based on our coupling $G$ is obtained from $G_1$ by adding the edges of $G_2$. Since $G_1 \sim G(n,d/n)$ and $d=\Theta(1)$, it is a standard application of the second-moment method that $|\mathcal{R}(G_1)|=\Theta(n)$ w.h.p. as $n \rightarrow +\infty$ (see Lemma \ref{lem:isol_nodes}). Using now that $G_2$ is independent with $G_1,$ $G_2$ on $\mathcal{R}(G_1)$ is an instance of $G(|\mathcal{R}(G_1)|,n^{-3/2+\epsilon}/(1-d/n))$. By Lemma \ref{lem:isol_edges} we conclude that there are $\Omega(n^{1/2+2\epsilon/3})$ edges of $G_2$ between vertices of $\mathcal{R}(G_1)$ that do not share a vertex with any other edge of $G_2$, w.h.p. as $n \rightarrow +\infty.$  As an outcome, there are $\Omega(n^{1/2+2\epsilon/3})$ vertices of $\mathcal{R}(G_1)$ that, while they are all used by all maximal independent sets of $G_1,$ they are now impossible to be used by any maximal independent set of $G=G_1 \cup G_2.$ Therefore, w.h.p. as $n \rightarrow +\infty,$
 $$\alpha(G_1)-\alpha(G) \geq \Omega( n^{1/2+2\epsilon/3})$$ and in particular,
 $$\PP_{d/n+n^{-3/2+\epsilon}} \left( \alpha(G_1)-\alpha(G) \leq 2n^{1/2+\epsilon/2} \right)=o(1).$$Combining it now with \eqref{eq:ineq} we conclude that for large enough $n$,
 $$\PP_{d/n+n^{-3/2+\epsilon}} \left(G \in \mathcal{I}_k\right) \leq \xi,$$as we wanted.

\end{proof}

\begin{proof}[Proof of Theorem \ref{thm:circuit_aon}]

 We start with a construction. Given the successful circuit $\mathcal{C}: \{0,1\}^{N} \rightarrow \{0,1\}^{N}$, we augment it to a Boolean circuit $\mathcal{C}':\{0,1\}^{2N} \rightarrow \{0,1\}$ of depth $d'=d+3$ and size $s' \leq c_0 (s+N),$ for some universal constant $c_0>0,$ as follows.
 
 We add $N$ fresh bits to the input $X,$ which will correspond to the Hidden Subset vector $S \in \{0,1\}^N$. Hence, the input of $\mathcal{C}'$ is $(X,S).$ Now we add three extra layers on top of the layers of $\mathcal{C}$ with total $O(s+N)$ new gates. The first two added layers are used to compute the bits $(\mathcal{C}(X)_i \land S_i) \lor (\neg \mathcal{C}(X)_i \land \neg S_i)$ (i.e., the indicator of $\mathcal{C}(X)_i = S_i$) for $i=1,\ldots,N.$ Note that these two layers can be clearly implemented with $O(s+N)$ gates. The third added layer computes the bit $\land_{i=1}^N \left((\mathcal{C}(X) \land S_i)\cup (\neg \mathcal{C}(X)_i \land \neg S_i)\right)$ (i.e., the indicator of $\mathcal{C}(X) = S$). This final gate can be implemented with one extra gate.

It is immediate to see that $\mathcal{C}'(X,S)=1$ if and only if $\mathcal{C}(X) = S$. Since $\mathcal{C}$ solves the Hidden Subset Problem indexed by $\mathcal{P}$ we conclude that for any $p < p_{\mathrm{IT}},$
\begin{align}\label{eq:large}
    \P_{S \sim \mathcal{P}, \bX \sim \P_p} ( \mathcal{C}'(\bX,S)=1) \geq 0.9.
\end{align} Moreover, as the Hidden Subset Problem indexed by $\mathcal{P}$ exhibits the All-or-Nothing (AoN) Phenomenon for window size $\epsilon=\epsilon(\mathcal{P}_N) \in (0,1)$, then for any $p>(1+\epsilon)p_{\mathrm{IT}}$ the circuit $\mathcal{C}: \{0,1\}^N \rightarrow \{0,1\}^N$ must satisfy
\begin{align*}
    \P_{S \sim \mathcal{P}, \bX \sim \P_p} ( \mathcal{C}(\bX)=S) \leq 0.1.
\end{align*} Using again $\mathcal{C}'(X,S)=1$ if and only if $\mathcal{C}(X) = S$, we conclude that for any $p>(1+\epsilon)p_{\mathrm{IT}}$,
\begin{align}\label{eq:small}
    \P_{S \sim \mathcal{P}, \bX \sim \P_p} ( \mathcal{C}'(\bX,S)=1) \leq 0.1.
\end{align}

Combining \eqref{eq:large} and \eqref{eq:small} we conclude 

\begin{align}\label{eq:difference_0}
   &\E_{S \sim \mathcal{P}} \left[\E_{p_{\mathrm{IT}}} \mathcal{C}'(\bX,S)-\E_{(1+\epsilon)p_{\mathrm{IT}}} \mathcal{C}'(\bX,S) \right] \\
   =&\E_{S \sim \mathcal{P}} \left[\P_{p_{\mathrm{IT}}} (\mathcal{C}'(\bX,S)=1)-\P_{(1+\epsilon)p_{\mathrm{IT}}} (\mathcal{C}'(\bX,S)=1) \right] \geq 0.8,
\end{align}where $\P_p$ is the $p$-biased measure on $\bX.$

This implies that for some specific $S$ in the support of $\mathcal{P},$

\begin{align}\label{eq:difference}
  \E_{p_{\mathrm{IT}}}\mathcal{C}'(\bX,S)-\E_{(1+\epsilon)p_{\mathrm{IT}}} \mathcal{C}'(\bX,S) \geq 0.8,
\end{align} 

Therefore, for this fixed $S$, the circuit $\mathcal{C}_S(X)=\mathcal{C}'(X,S)$ has a sharp threshold at $p_c=p_{\mathrm{IT}},$ of window $\epsilon,$ jump size $\delta=\Theta(1)$. Now, the depth of $\mathcal{C}_S$ equals
$$d_S=d'=d+3$$and its size is $$s_S=s' \leq c_0 s(\mathcal{C})+c_0N.$$Using then the sharp threshold of $\mathcal{C}_S$ and Theorem \ref{thm:main_1} we conclude for $\beta:=\min\{p_{\mathrm{IT}},1-p_{\mathrm{IT}}\}$ that for some constants $c_1,c_2,c_3>0,$ for large enough $N,$

 \begin{align}\label{eq:final_000}
      d(\mathcal{C}) \geq \frac{\log \frac{1-p_{\mathrm{IT}}}{\epsilon\log 1/\beta}}{2\log \log N \log 1/\beta}-7.
    \end{align} or
    \begin{align}\label{eq:final_01}
     s(\mathcal{C}) \geq c_2\exp \left( c_3(\frac{1-p_{\mathrm{IT}}}{\epsilon \log 1/\beta})^{1/(d(\mathcal(C)+6)}\right)
    \end{align}
\end{proof}
    
\begin{proof}[Proof of Lemma \ref{lem:sharp_transition_pc}]

To prove the AoN phenomenon, we show that for some $\epsilon=\Theta(1/k^2),$ if $p=(1-\epsilon)\binom{n}{k}^{-1/\binom{k}{2}},$ exact recovery is possible (the ``all'' part), but if $p=(1+\epsilon)\binom{n}{k}^{-1/\binom{k}{2}},$ exact recovery is impossible (the ``nothing'' part). Since for our range of $k$ clearly $\binom{n}{k}^{-1/\binom{k}{2}}=\exp(-\Theta(\log n/k))=1-\Theta(\log n/k)$ the proof follows.

  For notation purposes, we call $\mathcal{PC}$ the ``planted clique'' which naturally corresponds to the edges of the hidden subset $S$ in the (for now, interperatable as $n$-vertex undirected graph), $Y=S \lor \bX \in \{0,1\}^{\binom{n}{2}}$ where $\bX \sim \P_p$. We also definte $Z(Y)$ to be the random variable that counts the number of $k$-cliques in the graph $Y \in \{0,1\}^{\binom{n}{2}}$.  Also, for $\ell=0,1,\ldots,k-1$ call $Z_{\ell}(Y)$ the number of cliques of size $k$ in an instance of $Y$ which use exactly $\ell$ vertices of $\mathcal{PC}.$

    Notice that to establish first the ``all'' part, it suffices to show that for $p=(1-\epsilon)\binom{n}{k}^{-1/\binom{k}{2}},$ it holds 
    \begin{align*}
        \P_{S, \bX \sim \P_p}( \sum_{\ell=0}^{k-1} Z_{\ell}(Y) \geq 1) \leq 0.01.
    \end{align*}Indeed, this  implies that with probability at least $0.99,$ $\mathcal{PC}$ is the only $k$-clique in $Y$ and therefore exact recovery is possible by brute-force search for a $k$-clique in $Y$. Moreover, notice that by Markov's inequality and linearity of expectation, it suffices to prove,
    \begin{align*}
         \sum_{\ell=0}^{k-1} \E_{S, X \sim \P_p} Z_{\ell}(Y) \leq 0.01.
    \end{align*}

    Now for each $\ell>0,$ there are $\binom{k}{\ell}\binom{n-k}{k-\ell}$ 
    subsets of $[n]$ which correspond in $Y$ to $k$-vertex subgraphs with exactly $\ell$ vertices of $\mathcal{PC}.$ Moreover each of these subgraphs form a $k$-clique in $Y$ with probability $p^{\binom{k}{2}-\binom{\ell}{2}} \leq (1-\epsilon)^{\binom{k}{2}-\binom{\ell}{2}}\binom{n-k}{k}^{-1+\ell^2/k^2}.$ Hence, by linearity of expectation,
    $$\E_{S, X \sim \P_p} Z_{\ell}(Y)  \leq (1-\epsilon)^{\binom{k}{2}-\binom{\ell}{2}}\binom{k}{\ell}\binom{n}{k-\ell}\binom{n}{k}^{-1+\ell^2/k^2}.$$Therefore, it suffices to show \begin{align}\label{eq:goal}
         \sum_{\ell=0}^{k-1} (1-\epsilon)^{\binom{k}{2}-\binom{\ell}{2}}\binom{k}{\ell}\binom{n-k}{k-\ell}\binom{n}{k}^{-1+\ell^2/k^2} \leq 0.01.
    \end{align}

    For $\ell=0,$ notice
    \begin{align}\label{eq:goal_1}
        (1-\epsilon)^{\binom{k}{2}-\binom{\ell}{2}}\binom{k}{\ell}\binom{n-k}{k-\ell}\binom{n}{k}^{-1+\ell^2/k^2}=(1-\epsilon)^{\binom{k}{2}}\binom{n-k}{k}\binom{n}{k}^{-1}\leq (1-\epsilon)^{\binom{k}{2}}.
    \end{align}

We know that for some constant $\delta>0$ it holds $k \leq n^{1/3-\delta}.$ Let $\delta'>0$ such that 
\begin{align}\label{eq:delta'}
    (1-\delta)(1-\delta'/3)<1-2\delta'.
\end{align} For $1 \leq \ell \leq (1-\delta')k,$ using standard inequalities and that $k=o(n),$ we have for large enough $n$, 
\begin{align*}
    \binom{k}{\ell}\binom{n-k}{k-\ell}\binom{n}{k}^{-1}& \leq (\frac{ke}{\ell})^{\ell}\binom{n-k}{k-\ell}\binom{n}{k}^{-1} \\
     &\leq (\frac{ke}{\ell})^{\ell}\binom{n}{k-\ell}\binom{n}{k}^{-1} \\
    &= (\frac{ke}{\ell})^{\ell} \frac{k!(n-k)!}{(k-\ell)!(n-k+\ell)!}\\
    &\leq (e\frac{k^2}{\ell(n-2k)})^{\ell}.
\end{align*} Moreover,
\begin{align}\label{eq:binom}
    \binom{n}{k}^{\ell^2/k^2} \leq (\frac{ne}{k})^{\ell^2/k}.
\end{align} Hence for each $1 \leq \ell \leq (1-\delta')k,$

    $$\E_{S, X \sim \P_p} Z_{\ell}(Y)\leq (e^2\frac{k^2}{\ell(n-2k)} (\frac{n}{k})^{\ell/k})^{\ell}.$$But the function $F(\ell):=e\frac{k^2}{\ell(n-2k)} (\frac{n}{k})^{\ell/k}, \ell \in [1,(1-\delta')k]$ is log-convex (i.e., $\log F$ is convex), so for some constant $C>0,$ for large enough $n$ (using $(\log n)^{\omega(1)}=k=o(n^{1/3}),$ $$\max F(\ell) \leq \max\{F(1),F((1-\delta')k)\}=C\max\{ \frac{k^2}{n-2k}, (\frac{k}{n})^{-\delta'}\}.$$ Since $k =o(n^{1/3})=o(\sqrt{n}),$ we conclude for some constant $c_0>0$, for large enough $n$ that for each $1 \leq \ell \leq (1-\delta')k,$, $F(\ell) \leq n^{-c_0}.$ Hence,
    \begin{align}\label{eq:goal_2}
         \sum_{\ell=1}^{\floor{(1-\delta')k}} (1-\epsilon)^{\binom{k}{2}-\binom{\ell}{2}}\binom{k}{\ell}\binom{n-k}{k-\ell}\binom{n}{k}^{-1+\ell^2/k^2} \leq \sum_{\ell=1}^{\floor{(1-\delta')k}} n^{-c_0 \ell} \leq 2n^{-c_0}.
    \end{align}
    
    Finally, for $(1-\delta')k \leq \ell \leq k-1$ by standard inequalities, 
    \begin{align*}
        \binom{k}{\ell}\binom{n}{k-\ell}\binom{n}{k}^{-1+\ell^2/k^2}&=\binom{k}{k-\ell}\binom{n}{k-\ell}\binom{n}{k}^{-1+\ell^2/k^2}\\
        &\leq  (kn)^{k-\ell} (\frac{n}{k})^{-k+\ell^2/k} \\
        &\leq (kn (\frac{k}{n})^{1+\ell/k})^{k-\ell}\\
        &\leq (kn (\frac{k}{n})^{2-\delta'})^{k-\ell}=(k^{3-\delta'}/n^{1-\delta'})^{k-\ell}.
    \end{align*} By assumption on $k \leq n^{1/3-\delta}$ and the definition of $\delta'>0$ \eqref{eq:delta'} we have $$k^{3-\delta'} \leq n^{(1-\delta)(1-\delta'/3)} \leq n^{1-2\delta'.} $$ Hence,
    \begin{align}\label{eq:goal_3}
         \sum_{\ell=\floor{(1-\delta)k}}^{k-1} (1-\epsilon)^{\binom{k}{2}-\binom{\ell}{2}}\binom{k}{\ell}\binom{n-k}{k-\ell}\binom{n}{k}^{-1+\ell^2/k^2} &\leq \sum_{\ell=\floor{(1-\delta)k}}^{k-1} (k^3/n)^{k-\ell} \\
         &\leq \sum_{\ell=\floor{(1-\delta)k}}^{k-1} n^{-\delta'(k-\ell)} \\
         &\leq n^{-\delta'}.
    \end{align}Combining \eqref{eq:goal_1}, \eqref{eq:goal_2}, \eqref{eq:goal_3} we conclude
    \begin{align}\label{eq:goal_almost}
         \sum_{\ell=0}^{k-1} (1-\epsilon)^{\binom{k}{2}-\binom{\ell}{2}}\binom{k}{\ell}\binom{n-k}{k-\ell}\binom{n}{k}^{-1+\ell^2/k^2} \leq  (1-\epsilon)^{\binom{k}{2}}+o(1).
    \end{align}For some $\epsilon=\Theta(1/k^2)$ we conclude for large enough $n$,  \eqref{eq:goal} and the ``all'' part is proven.

    For part (b), we employ a variant of the planting trick \cite{achlioptas2008algorithmic} as used in statistical contexts \cite{coja2022statistical, mossel2023sharp}. First note that the posterior of $S$ given $Y$ is simply the uniform measure over the $k$-cliques in $Y$. Hence, for any specific instance of $Y$, exact recovery is not possible if $Y$ contains more than one $k$-clique. Therefore, since $Z(Y)$ counts the number of $k$-cliques in $Y$, it suffices to prove that for $p=(1+\epsilon)\binom{n}{k}^{-1/\binom{k}{2}},$ and large enough $n$, 
    \begin{align}\label{eq:goal_b}
        \P_{S, \bX \sim \P_p}(Z(Y=S \lor \bX)=1) \leq 0.01.
    \end{align}Now recall that $\P_p$ is the product measure on $\{0,1\}^{\binom{n}{2}},$ where all bits appear independently with probability $p.$ We compare the likelihood $Y$ is generated under our original measure versus the ``null'' case $Y=\bX \sim \P_p$ (i.e., where no Hidden Subset is considered). This gives for any $Y$, 
    \begin{align*}
        \frac{\P_{S, \bX \sim \P_p}(Y)}{\P_p(Y)}
        &=\sum_{S' \text{ k-subset}} \mathcal{P}(S') \frac{\P_p(Y|S' \text{ planted clique in }Y)}{\P_p(Y)}\\
        &=\sum_{S'} \binom{n}{k}^{-1} \frac{1(S \text{ k- clique in } Y)p^{|E(Y)|-\binom{k}{2}}(1-p)^{\binom{n}{2}-|E(Y)|}}{p^{|E(Y)|}(1-p)^{\binom{n}{2}-|E(Y)|}}\\
        &=\frac{Z(Y)}{\binom{n}{k}p^{\binom{k}{2}}}.
    \end{align*}
   
    But by our assumptions on $\epsilon>0$ we can take $\epsilon=\Theta(1/k^2)$ so that, $$\binom{n}{k}p^{\binom{k}{2}}=(1+\epsilon)^{\binom{k}{2}} \geq 100.$$Hence for this choice of $\epsilon>0$ and large enough $n$,
    \begin{align*}
        \P_{S, \bX \sim \P_p}(Z(Y) =1)& \leq \P_{S, \bX \sim \P_p}(Z(Y) \leq 0.01\binom{n}{k}p^{\binom{k}{2}})\\
        &=\E_{Y \sim P_p}\left[1(Z(Y) \leq 0.01\binom{n}{k}p^{\binom{k}{2}})\frac{P_{S, \bX \sim \P_p}(Y)}{P_p(Y)}\right]\\
        &=\E_{Y \sim P_p}\left[1(Z(Y) \leq 0.01\binom{n}{k}p^{\binom{k}{2}})\frac{Z(Y)}{\binom{n}{k}p^{\binom{k}{2}}}\right]\\
        &\leq 0.01\P_p(Z(Y) \leq 0.01 \binom{n}{k}p^{\binom{k}{2}})\leq 0.01.
    \end{align*}This completes the proof of the ``nothing'' part.

\end{proof}

\section{Auxilary lemmas}\label{sec:aux}
\begin{lemma}\label{lem:isol_edges}
    Suppose that for some $0<\epsilon<1/2,$ $G \sim G(n,p)$ for $p=\Theta(n^{-3/2+\epsilon}).$ Then there are $\Omega(n^{1/2+2\epsilon/3})$ edges of $G$ that do not share a vertex with any other edge of $G$, w.h.p. as $n \rightarrow +\infty.$ 
\end{lemma}

\begin{proof}
 Let $Z$ be the number of  edges of $G$ that do not share a vertex with any other edge of $G$. It holds $\E Z= n(n-1) p (1-p)^{2n-2}=\Theta(n^{1/2+\epsilon})$ and $\E Z^2=n(n-1) p (1-p)^{2n-2}+n(n-1)(n-2)(n-3)p^2(1-p)^{4n-4}=(1+o(1))(\E Z)^2.$ Hence by Chebychev's inequality we conclude that $Z \geq n^{-\epsilon/3}\E Z=\Omega(n^{1/2+2\epsilon/3}) $ w.h.p. as $n \rightarrow +\infty.$
\end{proof}
    
\begin{lemma}\label{lem:isol_nodes}
    Suppose that for some $d=\Theta(1),$ $G \sim G(n,d/n)$. Then there are $\Theta(n)$ isolated vertices, w.h.p. as $n \rightarrow +\infty.$ 
\end{lemma}

\begin{proof}
    Let $Z$ be the number of isolated vertices of $G$. Then by linearity of expectation $\E Z=n(1-d/n)^{n-1}=n(e^{-d}+o(1))$ and $\E Z^2=n(1-d/n)^{n-1}+n(n-1)(1-d/n)^{2n-3}=n^2(e^{-2d}+o(1))=(1+o(1))(\E Z)^2.$ Hence by Paley-Zygmund inequality we conclude $Z=\Theta(\E Z)=\Theta(n),$ w.h.p. as $n \rightarrow +\infty.$ 
\end{proof}

\section{Proof of the converse result Theorem \ref{thm:conv}}

% \dg{epsilons added inside $k$}

\subsection{Getting started}

Here we start with defining some notation that will be used throughout the proof. Let $\alpha>0$. For the proof, we choose $\zeta=\zeta(\alpha,C) \in (0,\alpha)$ to be a sufficiently small constant such that $$M:=(1-2\alpha)/\zeta \in \mathbb{Z}_+.$$ Importantly this choice allows us to partition $[p_{\alpha},p_{1-\alpha}]$ by choosing the points for $i=0,1,\ldots,M$
\begin{align}\label{eq:q_i}
q_i:=p_{\alpha+i\zeta}
\end{align}where recall that for any $\alpha' \in (0,1)$, we define $p_{\alpha'} \in (0,1)$ to be the unique value such that $\E_{p_{\alpha'}}f=\alpha'.$ 

Moreover, we highlight  a fact we are going to use throughout the proof  that all points in the ``threshold interval'' $[p_{\alpha},p_{1-\alpha}]$ are of the same order as $p_{1/2}$ which just for the proof of Theorem \ref{thm:conv} we abuse notation compared to the rest of the paper, and denote it as $p_c=p_{1/2}$. Specifically, according to the Bollobas-Thomason theorem \cite{bollobas1987threshold}, there exists $c_1=c_1(\alpha),C_1=C_1(\alpha)>0$ constants such that \begin{align}\label{eq:bol_thom}
   c_1 p_c \leq p_{\alpha} \leq p_{1-\alpha}\leq C_1p_c. 
\end{align}
\subsection{Proof sketch and two key lemmas}\label{sec:sketch}
The first step towards constructing the circuit for Theorem \ref{thm:conv} is to apply a celebrated result on coarse thresholds by Friedgut \cite[Theorem 1.1.]{Friedgut} which says that if a monotone Boolean function corresponding to a graph property has a coarse threshold (in the sense of \eqref{eq:deriv}) at any specific point of the threshold interval $p \in [p_{\alpha},p_{1-\alpha}]$, then it can be approximated with respect to the $p$-biased product measure by another monotone Boolean function whose minimal elements are all $O(1)$-size balanced graphs.

As for us $p \in [p_{\alpha},p_{1-\alpha}]$ is arbitrary and, importantly, we need to construct a circuit that is close to the Boolean function of interest for all values of $p$ in the threshold interval, we apply Friedgut's theorem on all the endpoints $q_i, i=0,1,\ldots,M$. Specifically, for any $\delta>0$ Friedgut's theorem \cite[Theorem 1.1.]{Friedgut} applied for $p=q_{i}$ for all $i=0,\ldots,M$ implies that for some $k(\gamma,C,\alpha)>0$ there exists \emph{a sequence} of monotone graph-inclusion properties (expressed as Boolean functions) $g_i: \{0,1\}^N \rightarrow \{0,1\}, i=0,\ldots,M$ such that the minimal elements of $g_i$ are balanced graphs of size at most $k(\gamma,C,\alpha)$ and
        \begin{align}\label{eq:fri}
        \max_{ i=0,1,\ldots,M}| \E_{q_i} f-\E_{q_i} g_i| \leq \gamma.
        \end{align}

        The key connection between \eqref{eq:fri} and Theorem \ref{thm:conv} is that exactly because all of the minimal elements of $g_i$ are of $k(\gamma,C,\alpha)=O(1)$-size, each of the $g_i, i=0,\ldots,M$ can be computed via an $\mathcal{AC}_0$ circuit of depth at most 2 and size at most $O(N^{k(\gamma,C,\alpha)})$ as one just needs to check whether at least one of the minimal elements of $g_i$ appears in the input $X.$

Now the sole fact that each of the $g_i$ is in $\mathcal{AC}_0$ is not sufficient for us to construct a successful circuit for Theorem \ref{thm:conv} as the input bit probability $p \in [p_{\alpha},p_{1-\alpha}]$ can be arbitrary (hence potentially not equal to any of the $q_i$'s) and on top of that $p$ is not given as input to the circuit and we have to have a rule to choose between the $q_i$'s. The strategy we follow to tackle these issues is statistical: the circuit we construct first produces a ``good'' estimate $\hat{p}$ for the bit bias $p$ based on the bits of the input graph $X$, and then outputs the corresponding $g_{i}$ for some $i$ such that $|\hat{p}-q_i|$ is small enough.

Notice that, albeit quite natural, implementing this idea is non-trivial because computing $\hat{p}$ and comparing it with the $q_i$'s should be computable via an $\mathcal{AC}_0$ circuit and standard statistical estimates such as averaging all bits of $X$ and thresholding the value are folklore known not to be in $\mathcal{AC}_0$ (e.g., see \cite{Smolensky} for a lower bound against computing the majority Boolean function, which is also an easy corollary of our main Theorem \ref{thm:main_1}). 

Before we describe how we accomplish this statistical estimation approach, also note that for this to be successful in proving Theorem \ref{thm:conv} some robustness is required in our guarantee \eqref{eq:fri}. More specifically, we first need to show that the guarantee \eqref{eq:fri} is robust to a small amount of error (both with respect to estimating the bias $p$, but also with respect to choosing the $g_i$ based on some $q_i$ close to $\hat{p}$). The following perturbation argument takes care of these latter concerns, which can also be seen as a slight strengthening of \eqref{eq:fri}. The proof is based on elementary but delicate random graph theory calculations and can be found in Section \ref{sec:pfs}.

\begin{lemma}\label{lem:max} Under the assumptions of Theorem \ref{thm:conv} and the above notation, for all $\gamma>0$, there exists $\zeta=\zeta(C,\alpha,\gamma)>0$ small enough, such that
   \begin{align}\label{eq:q_i_final}
\max_{i=0,1,\ldots,M}\sup_{q \in [q_i,q_{i+1}]} \max_{m \in \{i-1,i,i+1\}} | \E_{q }f-\E_{q} g_m| \leq \gamma. 
\end{align}

\end{lemma}

Now we explain how the circuit can also compute a good estimate for the value of $p$ based on $X$ in a way that can be implemented by an $\mathcal{AC}_0$ circuit. While one cannot naively apply majority on the whole input \cite{Smolensky} we show that a carefully constructed \emph{approximate} majority estimate using only the first $o(\log n)$ bits bypasses this difficulty and produces a sufficiently good estimate of $p$.

Motivated by \eqref{eq:bol_thom} we choose \begin{align}\label{eq:K}
    K:=\floor{1/p_c}=o(N),
\end{align}where that $K=o(N)=o(n^2)$ follows by the assumptions of the Theorem \ref{thm:conv}. We also choose an arbitrary $S=\omega(1)$ with $S=o(\log N)$ and $SK=o(N)$. Now consider the $S\times K$ first bits out of the $N$ total bits of $X$, and then consider the OR of each of the $S$ consecutive $K$ blocks of bits. Since $X \in \{0,1\}^N$ consists of independent bits with bit probability $p$, this creates $S$ independent bits of bias $b(p)=1-(1-p)^K$. Importantly for us notice that for any $p \in [p_{\alpha}, p_{1-\alpha}]$ and our choice of $K$, \eqref{eq:bol_thom} implies that the bias always satisfies for the $p$'s of interest that $b(p)=\Theta(1)$.

Moreover, we prove that since $\omega(1)=S=o(\log N)$ a majority estimate of $b(p)$ using the average of these $S$ independent bits can be computed with an $\mathcal{AC}_0$ circuit with error $1/\mathrm{poly}(S)=o(1)$. Then since $b(p)=\Theta(1)$ and for all $i$, $b(q_i)=\Theta(1)$, using \eqref{eq:bol_thom} we can also establish that for all $q_i$ ``slightly away'' from $p$ it must hold $|b(q_i)-b(p)|=\Omega(1)$; hence the error from our estimation allows to approximately identify the $i$ such that $|p-q_i|$ is minimized, completing the result. The above ideas are described formally in the following second key lemma, also proven in Section \ref{sec:pfs}.

\begin{lemma}\label{lem:counting}
 Under the assumptions of Theorem \ref{thm:conv} and the above notation, there exists a sufficiently small enough $\zeta=\zeta(\alpha)>0$ and a constant $C_0(\alpha,\zeta)>0$ such that for any $S=\omega(1)$ and $K$ given in \eqref{eq:K} the following hold with high probability. 
 
 For $m=0,1,\ldots,M-1$ consider the interval $$I_m:=(b(q_m)-C_0(\alpha,\zeta),b(q_{m+1})+C_0(\alpha,\zeta)).$$Then, for all $i=0,1,\ldots,M-1,$ and all $q \in [q_i,q_{i+1}]$, any $S$ independent bits with bias $b(q)$ satisfy that their sum denote by $Q$ satisfies $Q/S \in I_i$. Moreover, any $m=0,1,\ldots,M-1$ satisfying $Q/S \in I_m$ satisfy $m \in \{i-1,i,i+1\}.$
\end{lemma}

\subsection{Description of the circuit}\label{sec:circ}

Using the definitions and identical notation with the previous section, we describe the successful circuit $\mathcal{C}$ in $\mathcal{AC}_0$ for Theorem \ref{thm:conv} with input a graph $X \in \{0,1\}^N.$

We choose $S=\omega(1)$ but $S=o(\log n)$
so that $S/p_c=o(N)$, and any $\gamma>0$ and $\zeta=\zeta(\gamma)>0$ small enough such that both Lemmas \ref{lem:max} and \ref{lem:counting} hold. The circuit requires only four layers, and next to each step, we identify the layers of the circuit required to be involved.
\begin{enumerate}
 \item[(1$^{\text{st}}$)] In the \emph{first layer}, the circuit computes the OR of the first $S$ blocks of $K$ entries, i.e., the bits $\lor_{i=mK+1}^{(m+1)K} X_{i}, m=0,1,\ldots,S$.

    \item[(2$^{\text{nd}}$ \& 3$^{\text{rd}}$)] Let us denote the number of ones among the $S$ bits in the first layer by $Q$. In the \emph{second and third layers}, the circuit computes the value $Q$, in the sense that for all $k=0,1,\ldots,S$ in the third layer there is a gate which is $1$ if and only if $Q=k$. This can be done by constructing a second layer with a unique gate corresponding to any possible configuration among the $S$ bits and then the third layer is an appropriate choice of OR between the bits from the second layer. Both layers can be computed with $O(2^S)=N^{o(1)}$ gates.

    \item[(1$^{\text{st}}$ \& 2$^{\text{nd}}$)] Let $g_i$ defined in \eqref{eq:fri} for a given $\alpha, \gamma>0$. In parallel to the above, using only the \emph{first and second layer} the circuit computes in parallel for all $i=0,1,\ldots,M,$ the $\mathcal{AC}_0$ circuit for $g_i$ on the whole input $X.$ The additional gates used for this part are $O(N^{k(\gamma,C,\alpha)}),$ where $k(\gamma,C,\alpha)$ is also introduced in \eqref{eq:fri}.

\item[(4$^{\text{th}}$)] In the \emph{fourth layer}, the circuit outputs the OR of the output of the circuits $g_i$'s computed in the first two layers, but only for the $i=0,1,\ldots,M$ such that $Q/S \in (b_i-C_0(\epsilon,\zeta),b_{i+1}+C_0(\epsilon,\zeta))$ where $Q$ is identified in the third layer and $C_0(\epsilon,\zeta)$ is defined in Lemma \ref{lem:counting}. The additional gates used for this part are $O(S M)=O(S)$.
\end{enumerate}

Clearly the circuit described above is in $\mathcal{AC}_0$. 

\subsection{Putting it all together} Now, using the two key lemmas described in the section \ref{sec:sketch} (and proven in Section \ref{sec:pfs}), and the construction of the circuit in the Section \ref{sec:circ}, Theorem \ref{thm:conv} follows in the following straightforward manner.

\begin{proof}[Proof of Theorem \ref{thm:conv}]

Let $\gamma>0$ and $\zeta(\gamma,C,\alpha)>0$ sufficiently small such that both Lemmas \ref{lem:max} and \ref{lem:counting} holds.
Now let any $p \in [p_{\alpha},p_{1-\alpha}]$, and let $i=0,1,\ldots,M-1$ be unique index such that $p \in [q_i,q_{i+1}]$. Using Lemma \ref{lem:counting} we know that with high probability the circuit outputs 1 if and only if $g_m(X)=1$ for some $m \in \{i-1,i,i+1\}$. In particular, by conditioning on a with high probability event for $f(X) \not = \mathcal{C}(X)$ it must be that for some $m \in \{i-1,i,i+1\}$ it holds $f(X) \not = q_m(X)$. Hence, $$\max_{q \in [q_i,q_{i+1}]} |\E_q f-\E_q \mathcal{C}| \leq 3 \max_{q \in [q_i,q_{i+1}]} \max_{m\in \{i-1,i,i+1\}}|\E_q f-\E_q g_m|+o(1) \leq 3\gamma+o(1)$$where in the last inequality we used Lemma \ref{lem:max}. As $\gamma>0$ is arbitrary, the theorem follows. 

\end{proof}

\subsection{Proof of Key Lemmas}\label{sec:pfs}

\begin{proof}[Proof of Lemma \ref{lem:max}]

   Recall that as we explained in the notation section, in what follows we treat $\zeta=\zeta(\delta,C,\alpha) \in (0,\alpha)$ to be a sufficiently small constant for the purposes that follow, such that $M:=(1-2\alpha)/\zeta \in \mathbb{Z}$. 
   
   Now we need the following intermediate lemma.
\begin{lemma}\label{lem:interm} For all $\gamma>0$, there exists $\zeta=\zeta(\gamma)>0$ small enough, such that
    \begin{align*}\max_{i=0,1,\ldots,M-1}| \E_{q_{\min\{i+2,M\}}} g_i-\E_{q_{i}} g_i| \leq \gamma. \end{align*}
\end{lemma}

\begin{proof}
      First, note that for all $0 <p \leq p_{1-\alpha},$ by Russo-Margulis and Parseval identities, $$p \frac{d}{dp}\E_p f =\sum_H  |H|\hat{f}(H)^2 \geq \sum_{H: H \not = \emptyset}  \hat{f}(H)^2 =1-\E_p f \geq \alpha.$$ 
     
      Hence, by mean value theorem it must hold for all $0 \leq k \leq M,$ and all $\zeta>0$ small enough,
    \begin{align*}\zeta p_{1-\epsilon} \geq \alpha(p_{\alpha+k \zeta}-p_{\alpha+(k-1)\zeta}).
    \end{align*} 
    Using \eqref{eq:bol_thom} we conclude that for some $c''=c''(\alpha)>0$ it holds \begin{align}\label{eq:q_i_upp}\max_{0 \leq i \leq M-1} |q_{i+1}-q_i| = \max_{0 \leq k \leq (1-2\alpha)/\zeta} p_{\alpha+k \zeta}-p_{\alpha+(k-1)\zeta}\leq c''\zeta p_c.\end{align}

    %gives the result. \dg{what is $p_{M\epsilon}$?} \iz{The $p$ that our property has probability $M\epsilon.$}

    We fix one $i \in \{0,1,\ldots,M-1\}$ in what follows. Now recall by Friedgut's theorem $g_i$ corresponds to the property of inclusion of one out of at most $k(\gamma,C,\alpha)$ constant-sized balanced subgraphs
    in $G(n,p)$. Call them $H^{(i)}_j, j=1,\ldots,C_i$ for some $C_i \leq k(\gamma,C,\alpha)$. We can assume without loss of generality, that all balanced subgraphs $H^{(i)}_j,j=1,\ldots,C_i$ satisfy \begin{align}\label{eq:crit}
p_c=\Theta(p_c(H^{(i)}_j))=\Theta(n^{-v(H^{(i)}_j)/e(H^{(i)}_j)}), \end{align} where the last equality holds by standard random graph theory since $H^{(i)}_j$ is of constant order \cite{JansonBook}.
For the first equality, observe that if for some $H^{(i)}_j$ we have $p_c=o(p_c(H^{(i)}_j))$ or  $p_c=\omega(p_c(H^{(i)}_j))$ then one can slightly modify $g_i$ to be the Boolean function whose corresponding property has minimal elements the minimal elements of $g_i$ but without $H^{(i)}_j$. Notice that since $p_c=o(p_c(H^{(i)}_j))$ or  $p_c=\omega(p_c(H^{(i)}_j))$ for any $p=\Theta(p_c)$, $\E_p g_i$ remains identical after this modification up to a $o(1)$ additive error, and in particular \eqref{eq:fri} is still satisfied.
    %\dg{what do you mean wlog? Isn't this a fact?} \iz{no, some may have in principle $p_c(H^{(i)}_j)$ of difference order than $p_c$.}

    Recall the monotone coupling where an instance $G(n,q_{\min\{i+2,M\}})$ can be seen as the union of an instance of $G(n,q_i)$ and an instance $G(n,p)$ for $$\Delta p:=(q_{\min\{i+2,M\}}-q_i)/(1-q_i).$$ 
    Notice that by monotone coupling the quantity $\E_{q_{\min\{i+2,M\}}} g_i-\E_{q_{i}} g_i$ is at most the probability one out of the $H^{(i)}_j, j=1,\ldots,C_i$ appears in $G(n,q_{\min\{i+2,M\}})$ but not in $G(n,q_i).$ This can be upper bounded by the probability that for some $j=1,\ldots,C_i$ and some proper edge-subgraph $T$ of an $H^{(i)}_j$ (i.e., $e(H^{(i)}_j)>e(T)$) there exists a copy of $H^{(i)}_j$ in the complete graph such as first the instance of $G(n,q_i)$ contains a copy of $T$ and second the independent instance of $G(n,\Delta p)$ contains a copy of $H^{(i)}_j \setminus T$. By a direct union bound this probability is at most $$\sum_{j=1,\ldots,C_i, T \subset   H^{(i)}_j} n^{v(H^{(i)}_j)}q_i^{e(T)}(\Delta p)^{e(H^{(i)}_j)-e(T)}.$$

    Combining that $\max_{i=0,1\ldots,M} q_i \leq p_{1-\epsilon}=1-\Omega(1)$ and \eqref{eq:q_i_upp}, we can choose $\zeta>0$ small enough such that $\Delta p/p_c<1$. Moreover by \eqref{eq:crit}, $q_i=c'_3p_c=c_3n^{-v(H^{(i)}_j)/e(H^{(i)}_j)}$ for some constants $c_3,c'_3>0.$ Combining the above for some constants $c_4,c_5,c_6>0,$ the probability of interest equals,
    \begin{align}
    & c_4\sum_{j=1,\ldots,C_i, T \subset   H^{(i)}_j} n^{v(H^{(i)}_j)}(p_c)^{e(H^{(i)}_j)}(\Delta p/p_c)^{e(H^{(i)}_j)-e(T)} \\
    & \leq c_5 \Delta p/p_c\sum_{j=1,\ldots,C_i, T \subset   H^{(i)}_j} n^{v(H^{(i)}_j)}(p_c)^{e(H^{(i)}_j)}\\
    & \leq c_6 (q_{\min\{i+2,M\}}-q_i)/(p_c(1-q_i))=O(\zeta),
    \end{align}

    where in the last inequality we used that $H^{(i)}_j$ is of constant size and in the last equality \eqref{eq:q_i_upp}. The lemma follows by choosing $\zeta$ sufficiently small.
\end{proof}
%\dg{I did not follow the first inequality. Isn't $(q_{i+1}-q_i)/(p_c(1-q_i))$ raised to some power?.
%I could not follow the argument beyond this point} \iz{it can always be upper bounded by that quantity, because $p^2 \leq p$ etc.}

Notice that given the Lemma \ref{lem:max} and \eqref{eq:fri}, for sufficiently small $\zeta>0$ we have that for all $i,$

\begin{align}\label{eq:q_i_final}
\sup_{q \in [q_i,q_{i+1}]} \max_{m \in \{i-1,i,i+1\}} | \E_{q }f-\E_{q} g_m| \leq 7\gamma. 
\end{align}Indeed, first observe that by choosing $\zeta>0$ small enough, by Lemma \ref{lem:max} for any $i-1 \leq j_1 \leq j_2 \leq i+1,$
\begin{align}\label{eq:cor_lem}
    |\E_{q_{j_2} }g_m-\E_{q_{j_1}} g_m| \leq \gamma.
\end{align}and since $q_i=p_{\epsilon+i\zeta}$ by definition, 
\begin{align}\label{eq:cor_lem_2}
    |\E_{q_{j_2} }f-\E_{q_{j_1}} f|=(j_2-j_1)\zeta \leq 2\zeta \leq \gamma.
\end{align}
Combining \eqref{eq:cor_lem}, \eqref{eq:cor_lem_2} with \eqref{eq:fri}, for any $i=0,1,\ldots,M-1$, $q \in [q_i,q_{i+1}]$ and $i-1 \leq m \leq i+1$ it holds
\begin{align}\label{eq:approx}
 | \E_{q }f-\E_{q} g_m| \leq 7\gamma.
    \end{align}
    Indeed, by a series of triangle inequalities
\begin{align*}
    | \E_{q }f-\E_{q} g_m|
    &\leq | \E_{q }f-\E_{q_i} f|+| \E_{q_i }f-\E_{q_i} g_i|+|\E_{q_i }g_i-\E_{q_i} g_m|+|\E_{q_i }g_m-\E_{q} g_m|\\
    & \leq | \E_{q_{i+1} }f-\E_{q_i} f|+2| \E_{q_i }f-\E_{q_i} g_i|+|\E_{q_i }f-\E_{q_i} g_m|+|\E_{q_{i+1} }g_m-\E_{q_i} g_m|\\
    & \leq 4\gamma+|\E_{q_i }f-\E_{q_i} g_m|\\
    & \leq 4\gamma+|\E_{q_i }f-\E_{q_m} f|+| \E_{q_m }f-\E_{q_m} g_m|+|\E_{q_{m} }g_m-\E_{q_{i}} g_m|\\
    & \leq 7\gamma.
\end{align*}As $\gamma>0$ is arbitrary the above concludes the proof.

\end{proof}

\begin{proof}[Proof of Lemma \ref{lem:counting}]
    Notice that for any $ a \in [\alpha, 1-\alpha]$ and $0 \leq \zeta \leq 1-\alpha-a$, \eqref{eq:deriv} \ and a direct application of the mean value theorem implies
   
   \begin{align*}
      C(p_{a+\zeta}-p_{a}) \geq p_{a}\zeta \geq p_{\alpha} \zeta.
   \end{align*}By \eqref{eq:bol_thom} we conclude that for some constant $1>c'=c'(\alpha)>0,$ it must hold for any $ a \in [\alpha, 1-\alpha]$ and $0 \leq \zeta \leq 1-\alpha-a$,
   \begin{align}\label{eq:gap}
      p_{a+\zeta}-p_{a} \geq c' p_c \zeta.
   \end{align}Since by definition $q_i:=p_{\alpha+i\zeta}$,  by \eqref{eq:gap}, for all $0 \leq i \leq M-1,$ 

    \begin{align}\label{eq:q_i}
        |q_{i+1}-q_i| \geq  \min_{0 \leq k \leq (1-2\alpha)/\zeta-1} p_{\alpha+k \zeta}-p_{\alpha+(k-1)\zeta} \geq c' \zeta p_c.
    \end{align}

    Since we focus on $p \in [q_0,q_M]$, by \eqref{eq:bol_thom} all $q$ of interest satisfy $q=\Theta(p_c).$ Hence, the bias of the $S$ independent bits satisfies $$b(p)=1-e^{-Kp}+o(1)=1-e^{-p/p_c}+o(1)=\Theta(1).$$Moreover, notice that $b(p)$ is increasing with $p$.

By \eqref{eq:q_i} we know that for all $i$ $q_{i+1}-q_i \geq c' \zeta p_c$ and therefore for all $i=0,\ldots,M-1$, using \eqref{eq:bol_thom} again, there exists a $0<C_0(\epsilon,\zeta)<b_0/2$ such that \begin{align*}
    b(q_{i+1})-b(q_i)&=e^{-q_i/p_c}(1-e^{-(q_{i+1}-q_i)/p_c})+o(1) \\
    &\geq e^{-p_{1-\epsilon}/p_c}(1-e^{-c' \zeta})+o(1) \\
    &\geq 2C_0(\epsilon,\zeta)=\Theta(1).
\end{align*}In particular, by the law of large numbers, for any $q \in [q_i,q_{i+1}]$ the probability that the $S=\omega(1)$ independent bits with bias  $b(q)$ have average value outside $(b(q_i)-C_0(\epsilon,\zeta),b(q_{i+1})+C_0(\epsilon,\zeta))$ is $o(1).$  From this, both conclusions of the lemma directly follow.
    
\end{proof}

\section{Acknowledgments}

The authors are thankful to Nike Sun, Manolis Zampetakis, Fotis Iliopoulos, Rahul Santhanam and Moshe Vardi for helpful discussions. D.G. was supported during this project by NSF grants CISE-2233897 and DMS-2015517. E.M. and I.Z. were supported during this project by the Vannevar Bush Faculty Fellowship ONR-N00014-20-1-2826, Theoretical Foundations of Deep Learning (NSF DMS-2031883.) and the Simons Investigator award (622132).

   \newpage
\bibliographystyle{alpha}
\bibliography{bibliography-06.2023,pc}

\end{document}